\documentclass[11pt,letterpaper,reqno]{amsart}

\title[]{Approximation of Semiclassical Expectation Values\\by Symplectic Gaussian Wave Packet Dynamics}
\author{Tomoki Ohsawa}
\address{Department of Mathematical Sciences, The University of Texas at Dallas, 800 W Campbell Rd, Richardson, TX 75080-3021}
\email{tomoki@utdallas.edu}

\date{\today}



\usepackage{graphicx,amssymb,mathrsfs,paralist}
\usepackage[margin=1in, marginpar=.5in]{geometry}

\usepackage[numbers,sort&compress]{natbib}
\usepackage[colorlinks=false]{hyperref}

\usepackage{tikz-cd}

\theoremstyle{plain}
\newtheorem{theorem}{Theorem}[section]
\newtheorem{corollary}[theorem]{Corollary}
\newtheorem{lemma}[theorem]{Lemma}
\newtheorem{proposition}[theorem]{Proposition}

\theoremstyle{definition}

\theoremstyle{remark}
\newtheorem{remark}[theorem]{Remark}

\def\od#1#2{\dfrac{d#1}{d#2}}
\def\pd#1#2{\dfrac{\partial #1}{\partial #2}}
\def\tpd#1#2{\partial #1/\partial #2}

\def\parentheses#1{{\left(#1\right)}}
\def\brackets#1{{\left[#1\right]}}
\def\braces#1{{\left\{#1\right\}}}

\def\tr{\mathop{\mathrm{tr}}\nolimits}

\def\norm#1{{\left\|#1\right\|}}
\def\abs#1{{\left|#1\right|}}
\def\DS{\displaystyle}
\def\R{\mathbb{R}}
\def\C{\mathbb{C}}
\def\N{\mathbb{N}}
\def\defeq{\mathrel{\mathop:}=}

\def\setdef#1#2{{\left\{ #1 \ |\ #2 \right\}}}
\def\ip#1#2{{\left\langle#1,#2\right\rangle}}

\def\bigip#1#2{{\bigl\langle#1,#2\big\rangle}}
\def\exval#1{{\left\langle#1\right\rangle}}

\def\bigexval#1{{\bigl\langle#1\bigr\rangle}}

\renewcommand{\Re}{\operatorname{Re}}

\def\eps{\varepsilon}
\def\Mat{\mathsf{M}}
\def\SO{\mathsf{SO}}

\def\rmi{{\rm i}}

\begin{document}

\footskip=.6in

\begin{abstract}
  This paper concerns an approximation of the expectation values of the position and momentum of the solution to the semiclassical Schr\"odinger equation with a Gaussian as the initial condition.
  Of particular interest is the approximation obtained by our symplectic/Hamiltonian formulation of the Gaussian wave packet dynamics that introduces a correction term to the conventional formulation using the classical Hamiltonian system by Hagedorn and others.
  The main result is a proof that our formulation gives a higher-order approximation than the classical formulation does to the expectation value dynamics under certain conditions on the potential function.
  Specifically, as the semiclassical parameter $\varepsilon$ approaches $0$, our dynamics gives an $O(\varepsilon^{3/2})$ approximation of the expectation value dynamics whereas the classical one gives an $O(\varepsilon)$ approximation.
\end{abstract}

\maketitle

\section{Introduction}
\subsection{Semiclassical Schr\"odinger Equation and Gaussian Wave Packet}
Consider the following initial value problem of the semiclassical Schr\"odinger equation on $\R^{d}$:
\begin{subequations}
  \label{ivp:Schroedinger}
  \begin{gather}
    \label{eq:Schroedinger}
    \rmi\,\eps \pd{}{t}\psi(t,x) = \hat{H} \psi(t,x),
    \qquad
    \hat{H} = \frac{\hat{p}^{2}}{2} + V(x), \\
    \label{eq:Schroedinger-IC}
    \psi(0,x) = \phi_{0}(q(0), p(0), Q(0), P(0), S(0); x),
    \end{gather}
\end{subequations}
where $\eps > 0$ is the semiclassical parameter, $\hat{p} \defeq -\rmi\eps\tpd{}{x}$ is the momentum operator, and $\phi_{0}$ is the Gaussian wave function
\begin{equation}
  \label{eq:phi_0}
  \phi_{0}(q, p, Q, P, S; x)
  \defeq \frac{ (\det Q)^{-1/2} }{ (\pi\eps)^{d/4} } \exp\braces{ \frac{\rmi}{\eps}\parentheses{ \frac{1}{2}(x - q)^{T}P Q^{-1}(x - q) + p \cdot (x - q) + S } }.
\end{equation}
The parameters $(q,p)$ live in the cotangent bundle $T^{*}\R^{d} \cong \R^{d} \times \R^{d}$, whereas $Q, P \in \Mat_{d}(\mathbb{C})$ (the set of $d \times d$ complex matrices) satisfy
\begin{equation*}
  Q^{T}P - P^{T}Q = 0
  \quad\text{and}\quad
  Q^{*}P - P^{*}Q = 2\rmi I,
\end{equation*}
and $S \in \R$ is a phase factor.
It is worth noting that the imaginary part of $PQ^{-1}$ is given by $(Q Q^{*})^{-1}$; see, e.g., \citet[Lemma~V.1.1]{Lu2008}.

The seminal works by \citet{Ha1980,Ha1998} (see also \citet{He1975a,He1976b,Heller-LesHouches}, \citet{Ro2007}, \citet{CoRo2012}) showed that one may approximate the solution $\psi(t,x)$ of the above initial value problem~\eqref{ivp:Schroedinger} in the semiclassical limit $\eps \to 0$ by the time-dependent Gaussian wave packet
\begin{equation}
  \label{eq:phi_0-t}
  \phi_{0}(t,x) \defeq \phi_{0}(q(t), p(t), Q(t), P(t), S(t); x)
\end{equation}
whose parameters evolve in time according to the ordinary differential equations
\begin{equation}
  \label{eq:Hagedorn}
  \begin{array}{c}
    \DS
    \dot{q} = p,
    \qquad
    \dot{p} = -DV(q),
    \qquad
    \dot{Q} = P,
    \qquad
    \dot{P} =  -D^{2}V(q)\,Q,
    \medskip\\
    \DS
    \dot{S} = \frac{p^{2}}{2} - V(q),
  \end{array}
\end{equation}
where $DV$ and $D^{2}V$ stand for the gradient and Hessian of the potential $V$.
Note that the equations for $(q,p)$ are the \textit{classical} Hamiltonian system.
Specifically, \citet{Ha1980,Ha1998} proved, under certain conditions on the potential $V$, that the error in terms of the $L^{2}$-norm $\norm{\,\cdot\,}$ in $L^{2}(\R^{d})$ is $O(\eps^{1/2})$ in the sense that there exists a function $\mathscr{C}(t)$ such that
\begin{equation}
  \label{eq:estimate-Hagedorn}
  \norm{ \psi(t,x) - \phi_{0}(t,x) } \le \mathscr{C}(t)\,\eps^{1/2}.
\end{equation}

\subsection{Symplectic Gaussian Wave Packet Dynamics}
In a series of works \cite{OhLe2013,Oh2015c,OhTr2017}, we proposed the following symplectic/Hamiltonian alternative to the evolution equation~\eqref{eq:Hagedorn}:
\begin{subequations}
  \label{eq:Symp_Hagedorn}
  \begin{gather}
    \label{eq:Symp_Hagedorn-qpQP}
    \DS
    \dot{q} = p,
    \qquad
    \dot{p} = -\partial_{q} \parentheses{ V(q) + \eps\,V^{(1)}(q,Q) },
    \qquad
    \dot{Q} = P,
    \qquad
    \dot{P} = -D^{2}V(q)\,Q,
    \medskip\\
    \label{eq:Symp_Hagedorn-S}
    \DS \dot{S} = \frac{p^{2}}{2} - V(q),
  \end{gather}
\end{subequations}
where $\partial_{q}$ is a shorthand for $\tpd{}{q}$ and
\begin{equation}
  \label{eq:V^1}
  V^{(1)}(q,Q) \defeq \frac{1}{4}\tr\parentheses{ Q Q^{*} D^{2}V(q) }.
\end{equation}
The only difference from \eqref{eq:Hagedorn} of Hagedorn is the $O(\eps)$ correction term in the potential.
The correction term renders the coupled system~\eqref{eq:Symp_Hagedorn-qpQP} for $(q, p, Q, P)$ a Hamiltonian system on $T^{*}\R^{d} \times \Mat_{d}(\C) \times \Mat_{d}(\C)$ with a natural symplectic structure and the following Hamiltonian~\cite{Oh2015c}:
\begin{equation}
  \label{eq:H}
  H^{\eps}(q,p,Q,P) \defeq \frac{p^{2}}{2} + V(q)
  + \frac{\eps}{4} \parentheses{ \tr(P^{*}P) + \tr\parentheses{ Q Q^{*}\, D^{2}V(q) } }.
\end{equation}
Note also that \eqref{eq:Symp_Hagedorn-S} is decoupled and $t \mapsto S(t)$ is obtained by a quadrature using the solution to \eqref{eq:Symp_Hagedorn-qpQP}.
In what follows, the time dependent functions $t \mapsto (q(t),p(t),Q(t),P(t),S(t))$ refer to the solution to \eqref{eq:Symp_Hagedorn} with the initial condition $(q(0),p(0),Q(0),P(0),S(0))$ at $t = 0$, unless otherwise stated.

This is in contrast to \eqref{eq:Hagedorn}, which is Hamiltonian in the decoupled classical dynamics of $(q,p)$ in the classical phase space $T^{*}\R^{d}$ but not as a system for $(q, p, Q, P)$.
We also note in passing that \citet{WaLuWe2017} obtained \eqref{eq:Symp_Hagedorn} from a different perspective, and also that a correction term of the above form was proposed earlier by \citet{PaSc1994} for the one-dimensional case, and also by \citet{PrPe2000,PrPe2002} and \citet{Pr2006} in a different manner.
Our formulation gave a symplectic-geometric account of the variational formulation of \citet{FaLu2006} (see also \citet[Section~II.4]{Lu2008}), and yielded the correction term as a result of an asymptotic expansion of the resulting potential term~\cite{OhLe2013,Oh2015c}.

\subsection{Main Result}
The main question we would like to address is how the $O(\eps)$ correction term in \eqref{eq:Symp_Hagedorn} contributes to the accuracy of the approximation by the Gaussian wave packet dynamics.
We are particularly interested in approximating the dynamics of the expectation values of the position and momentum operators
\begin{equation*}
  \hat{z} \defeq (\hat{x}, \hat{p}) = \parentheses{ \hat{x},-\rmi\eps\pd{}{x} },
\end{equation*}
that is, the dynamics defined as
\begin{equation}
  \label{eq:exval_dynamics}
  t \mapsto \exval{\hat{z}}(t) = \ip{ \psi(t,\,\cdot\,) }{ \hat{z}\,\psi(t,\,\cdot\,) },
\end{equation}
where $\ip{\,\cdot\,}{\,\cdot\,}$ is the standard (right-linear) inner product on $L^{2}(\R^{d})$, and $t \mapsto \psi(t,\,\cdot\,)$ is the solution of the initial value problem~\eqref{ivp:Schroedinger}.

Numerical experiments~\cite{OhTr2017} suggest that our dynamics~\eqref{eq:Symp_Hagedorn} gives a better approximation than the classical dynamics~\eqref{eq:Hagedorn} does to the exact expectation value dynamics.
We note that the comparisons were made with respect to the expectation value dynamics obtained by Egorov's method~\cite{CoRo2012,LaRo2010,LaLu2020} or the Initial Value Representation (IVR) method~\cite{Mi1970,Mi1974b,WaSuMi1998,Mi2001}, which is known to give an $O(\eps^{2})$ approximation to the exact dynamics~\eqref{eq:exval_dynamics}; see, e.g., \citet{Eg1969}, \citet{BoRo2002}, and \citet[Chapter~11]{Zw2012}.

Our main result gives a rigorous account of this observation:
\begin{theorem}
  \label{thm:main}
  Suppose that $V \in C^{4}(\R^{d})$ is bounded from below, i.e., $C_{1} \le V(x)$ for some $C_{1} \in \R$ for any $x \in \R^{d}$, and also that $x \mapsto D^{2}_{ij}V(x)$, i.e., the $(i,j)$-component of the Hessian $D^{2}V(x)$, is bounded for any $i,j \in \{1, \dots, d\}$.
  Let $t \mapsto (z(t), Q(t), P(t))$ with $z(t) \defeq (q(t), p(t))$ be the solution to \eqref{eq:Symp_Hagedorn-qpQP}, $t \mapsto z^{0}(t) = (q^{0}(t), p^{0}(t))$ be that to the classical Hamiltonian system in \eqref{eq:Hagedorn} with the initial condition $z^{0}(0) = z(0)$, and let $t \mapsto \exval{\hat{z}}(t)$ be the exact expectation value dynamics~\eqref{eq:exval_dynamics}.
  Then, for any $i \in \{1, \dots, 2d\}$, $z_{i}(t) - \exval{\hat{z}_{i}}(t) = O(\eps^{3/2})$ in the sense that there exists a function $\mathscr{C}_{i}(t)$ such that
  \begin{equation*}
    \abs{ z_{i}(t) - \exval{\hat{z}_{i}}(t) } \le \mathscr{C}_{i}(t)\,\eps^{3/2},
  \end{equation*}
  whereas $z^{0}_{i}(t) - \exval{\hat{z}_{i}}(t) = O(\eps)$ in the same sense.
\end{theorem}

Several remarks are in order.
This paper is \textit{not} about improving the accuracy of an ansatz for the wave function itself.
In fact, replacing \eqref{eq:Hagedorn} by \eqref{eq:Symp_Hagedorn} does \textit{not} improve the estimate \eqref{eq:estimate-Hagedorn} in terms of $\eps$, as we shall show in Corollary~\ref{cor:magic_formula} (with $\mathbf{n} = 0$).
The focus of the paper is rather on improving the approximation of the expectation value dynamics \textit{without having any additional time evolution equations other than \eqref{eq:Hagedorn} nor assuming any ansatz other than the Gaussian~\eqref{eq:phi_0-t}}: We achieve it by simply introducing a correction term to \eqref{eq:Hagedorn}.

We also would like to stress that, as shown in \cite{OhLe2013,Oh2015c}, the derivation of \eqref{eq:Symp_Hagedorn} also involves only the Gaussian \eqref{eq:phi_0}.
In other words, the $O(\eps)$ correction term does \textit{not} come from any higher-order ansatz as one might expect.
This is in contrast to assuming a higher-order ansatz than just the Gaussian~\eqref{eq:phi_0-t} as is done in \cite{Ha1981,Hagedorn1985,Ha1998,HaJo1999,HaJo2000} and \cite[Theorem~24 on p.~109]{CoRo2012}.
One can certainly improve the accuracy of the ansatz that way, but needs additional evolution equations in addition to those for $(q,p,Q,P,S)$.

Note also that our approximation involves only a single initial value problem of \eqref{eq:Symp_Hagedorn} as opposed to averaging solutions over numerous initial conditions like the Egorov/IVR method mentioned above.

Throughout the paper, we will carry out asymptotic analysis as $\eps \to 0$ of time-dependent functions, and will employ the same notation used in the statement of the above theorem for brevity.
Specifically, when we write $f(t,\eps) = O(\eps^{r})$ for some time-dependent function $f$ with some $r \in \R$, it means that there exists a function $\mathscr{C}(t)$ such that $\abs{ f(t,\eps) } \le \mathscr{C}(t)\,\eps^{r}$ as $\eps \to 0$.

\subsection{Approximation of Other Observables}
One also naturally wonders whether Theorem~\ref{thm:main} extends to the expectation values of general observables as well.
We defer this question to future work.
However, it is easy to see that the result holds for the Hamiltonian with an even better approximation:
As we mentioned above, our dynamics \eqref{eq:Symp_Hagedorn} is a Hamiltonian system with the Hamiltonian $H^{\eps}$ given in \eqref{eq:H}.
But then this Hamiltonian is an $O(\eps^{2})$ approximation to the expectation value of the Hamiltonian operator $\hat{H}$ from \eqref{eq:Schroedinger} with respect to the Gaussian:
\begin{equation*}
  \ip{ \phi_{0}(q,p,Q,P,S) }{ \hat{H} \phi_{0}(q,p,Q,P,S) } = H^{\eps}(q,p,Q,P) + O(\eps^{2}).
\end{equation*}
This follows from Laplace's method (see, e.g., \citet[Section~3.7]{Mi2006}) applied to the integral on the left.
Now, note that $t \mapsto \bigexval{ \hat{H} }(t) \defeq \bigip{ \psi(t) }{ \hat{H} \psi(t) }$ along the exact solution to \eqref{ivp:Schroedinger} and $t \mapsto H^{\eps}(q(t),p(t),Q(t),P(t))$ along \eqref{eq:Symp_Hagedorn} are both constant.
So their difference is constant at the initial value---where $\psi(0)$ is the initial Gaussian~\eqref{eq:Schroedinger-IC}:
\begin{equation*}
  \bigexval{ \hat{H} }(t) - H^{\eps}(q(t),p(t),Q(t),P(t))
  = \bigl\langle \hat{H} \bigr\rangle(0) - H^{\eps}(q(0),p(0),Q(0),P(0))
  = O(\eps^{2}).
\end{equation*}
On the other hand, with the classical system~\eqref{eq:Hagedorn} and the classical Hamiltonian $H^{0}(q,p) \defeq p^{2}/2 + V(q)$,
\begin{equation*}
  \bigexval{ \hat{H} }(t) - H^{0}(q^{0}(t),p^{0}(t))
  = \bigl\langle \hat{H} \bigr\rangle(0) - H^{0}(q(0),p(0))
  = O(\eps),
\end{equation*}
because $H^{\eps}(q,p,Q,P) = H^{0}(q,p) + O(\eps)$.

A similar argument works for, e.g., the angular momentum $J(q,p) \defeq q \diamond p$---where $q \diamond p$ denotes the $d \times d$ skew-symmetric matrix defined by $(q \diamond p)_{ij} \defeq q_{j}p_{i} - q_{i}p_{j}$---when the potential $V$ has $\SO(d)$-symmetry by approximating the expectation value $\bigexval{ \hat{J} }$ by the semiclassical angular momentum~\cite{Oh2015b,Oh2015c}:
\begin{equation*}
  J^{\eps}(q, p, Q, P) = q \diamond p + \frac{\eps}{2}\Re(P Q^{*} - Q P^{*}),
\end{equation*}
because both $\bigexval{ \hat{J} }$ and $J^{\eps}$ are invariants.

\subsection{Outline}
\label{ssec:outline}
We prove Theorem~\ref{thm:main} in the rest of the paper.
The main part of the proof is in Section~\ref{sec:proof}, whereas Sections~\ref{sec:HWP} and \ref{sec:time_evolution-HWP} are devoted to some lemmas and propositions needed in Section~\ref{sec:proof}.
Therefore, the reader might want to first skim through Section~\ref{sec:proof} to have an overview of the proof.

Much of what we do is a detailed analysis of the error $\mathcal{Z}_{0}(t;x) \defeq \psi(t,x) - \phi_{0}(t,x)$, i.e., the difference between the exact solution to \eqref{ivp:Schroedinger} and the Gaussian wave packet~\eqref{eq:phi_0-t}.
In fact, the difference between the exact expectation value $\exval{\hat{x}}(t)$ of the position and the position variable $q(t)$ in \eqref{eq:Symp_Hagedorn-qpQP} is, dropping the spatial variables $x$ for brevity,
\begin{align*}
  \exval{\hat{x}}(t) - q(t)
  &= \ip{\psi(t)}{(\hat{x} - q(t)) \psi(t)} \\
  &= \ip{\phi_{0}(t)}{(\hat{x} - q(t))\phi_{0}(t)} + \ip{\phi_{0}(t)}{(\hat{x} - q(t)) \mathcal{Z}_{0}(t)} + \ip{\mathcal{Z}_{0}(t)}{(\hat{x} - q(t)) \phi_{0}(t)} \\
  &\quad + \ip{\mathcal{Z}_{0}(t)}{(\hat{x} - q(t))\mathcal{Z}_{0}(t)}\\
  &= 2\Re\ip{\mathcal{Z}_{0}(t)}{(\hat{x} - q(t)) \phi_{0}(t)} + \ip{\mathcal{Z}_{0}(t)}{(\hat{x} - q(t)) \mathcal{Z}_{0}(t)}.
\end{align*}
Therefore, our analysis boils down to estimates of the above two terms involving the error $\mathcal{Z}_{0}$.
Those lemmas and propositions in Sections~\ref{sec:HWP} and \ref{sec:time_evolution-HWP} mainly concern those key properties of $\mathcal{Z}_{0}$ that are pertinent to our analysis.

\section{Hagedorn Wave Packets}
\label{sec:HWP}
\subsection{Overview}
We first give a brief review of the Hagedorn wave packets~\cite{Ha1980, Ha1981, Hagedorn1985, Ha1998} (see also \cite{CoRo2012} and \cite{Oh2019b}), and then derive the evolution equation satisfied by the Gaussian wave packet~\eqref{eq:phi_0-t} where the parameters evolve in time according to our equations~\eqref{eq:Symp_Hagedorn}.
The evolution equation resembles the Schr\"odinger equation~\eqref{eq:Schroedinger} but differs by a residual term.
We then prove several key properties of the residual term.
Later, in Section~\ref{sec:time_evolution-HWP}, we will find an expression for the error $\mathcal{Z}_{0}(t,x)$ in terms of the residual term analyzed here.

\subsection{The Hagedorn Wave Packets}
Following \citet{Ha1998}, let us define the lowering operator
\begin{equation*}
  \mathscr{A}(q,p,Q,P) \defeq -\frac{\rmi}{\sqrt{2\eps}} \parentheses{ P^{T}(\hat{x} - q) - Q^{T}(\hat{p} - p) }
\end{equation*}
as well as its adjoint or the raising operator
\begin{equation}
  \label{eq:Astar}
  \mathscr{A}^{*}(q,p,Q,P) \defeq \frac{\rmi}{\sqrt{2\eps}} \parentheses{ P^{*}(\hat{x} - q) - Q^{*}(\hat{p} - p) }.
\end{equation}
We refer to the lowering and raising operators collectively as the \textit{ladder operators}.
Note that both are operators on the Schwartz space $\mathscr{S}(\R^{d})$.
It is straightforward to see that they satisfy the following relationship for any $j, k \in \{1, \dots, d\}$:
\begin{equation}
  \label{eq:commutator-As}
  [\mathscr{A}_{j}(q,p,Q,P), \mathscr{A}^{*}_{k}(q,p,Q,P)] = \delta_{jk}.
\end{equation}
It also turns out to be convenient to write the position and momentum operators in terms of the ladder operators as follows:
\begin{equation}
  \label{eq:hatx-q_in_ladder_ops}
  \hat{x} - q = \sqrt{\frac{\eps}{2}}\parentheses{ \overline{Q} \mathscr{A}(q,p,Q,P) + Q \mathscr{A}^{*}(q,p,Q,P) },
\end{equation}
\begin{equation}
  \label{eq:hatp-p_in_ladder_ops}
  \hat{p} - p = \sqrt{\frac{\eps}{2}}\parentheses{ \overline{P} \mathscr{A}(q,p,Q,P) + P \mathscr{A}^{*}(q,p,Q,P) }.
\end{equation}

Let $\N_{0}$ be the set of integers greater than or equal to zero.
Then one can generate a set of functions $\{ \phi_{\mathbf{n}}(q,p,Q,P,S) \}_{\mathbf{n} \in \N_{0}^{d}}$ by recursively defining, for any multi-index $\mathbf{n} = (n_{1}, \dots, n_{d}) \in \N_{0}^{d}$ and $j \in \{1, \dots, d\}$,
\begin{equation}
  \label{eq:phi_n-raised}
  \phi_{\mathbf{n} + \mathbf{e}_{j}}(q,p,Q,P,S;x) \defeq \frac{1}{\sqrt{n_{j} + 1}}\,\mathscr{A}^{*}_{j}(q,p,Q,P) \phi_{\mathbf{n}}(q,p,Q,P,S;x),
\end{equation}
where $\mathbf{e}_{j}$ is the unit vector in $\R^{d}$ whose $j$-th entry is 1; then they also satisfy
\begin{equation*}
  \phi_{\mathbf{n} - \mathbf{e}_{j}}(q,p,Q,P,S;x) \defeq \frac{1}{\sqrt{n_{j}}}\,\mathscr{A}_{j}(q,p,Q,P) \phi_{\mathbf{n}}(q,p,Q,P,S;x).
\end{equation*}

\citet{Ha1980, Ha1981, Hagedorn1985, Ha1998} showed that the set $\{ \phi_{\mathbf{n}}(q,p,Q,P,S) \}_{\mathbf{n} \in \N_{0}^{d}}$ then forms an orthonormal basis for $L^{2}(\R^{d})$, where $\phi_{0}$ is the ``ground state'' here in the sense that
\begin{equation}
  \label{eq:A-phi_0}
  \mathscr{A}(q,p,Q,P) \phi_{0}(q,p,Q,P,S;x)  = 0.
\end{equation}
In fact, one may think of them as a generalization of the Hermite functions, and also can find a unitary operator on $L^{2}(\R^{d})$ that relates each element of the Hagedorn wave packet with the Hermite function of the same index~\cite{Oh2019b}.
Because of this correspondence, we refer to $\phi_{\mathbf{n}}$ as an $|\mathbf{n}|$-th \textit{excited state} with $|\mathbf{n}| \defeq \sum_{i=1}^{d}n_{i}$ for any multi-index $\mathbf{n} \in \N_{0}^{d}$.

We note in passing that \citet{Ha1998} actually constructed an orthonormal basis $\{ \varphi_{\mathbf{n}} \}_{\mathbf{n} \in \N_{0}^{d}}$ without the phase factor starting with the Gaussian
\begin{equation*}
  \varphi_{0}(q, p, Q, P; x) \defeq \frac{(\det Q)^{-1/2}}{(\pi\eps)^{d/4}} \exp\braces{ \frac{{\rm i}}{\eps}\parentheses{ \frac{1}{2}(x - q)^{T}P Q^{-1}(x - q) + p \cdot (x - q) } }
\end{equation*}
instead of $\phi_{0}$ from \eqref{eq:phi_0}.
It is just a matter of convenience that we use the basis $\{ \phi_{\mathbf{n}} \}_{\mathbf{n} \in \N_{0}^{d}}$ with the phase factor instead of $\{ \varphi_{\mathbf{n}} \}_{\mathbf{n} \in \N_{0}^{d}}$.

\subsection{Evolution Equation of the Gaussian Wave Packet}
\citet{Ha1980, Ha1981, Hagedorn1985, Ha1998} and \citet{HaJo1999,HaJo2000} have proved error estimates of various approximations to the solution to the Schr\"odinger equation~\eqref{eq:Schroedinger} constructed by taking linear combinations of the Hagedorn wave packets.
In their works, each wave packet evolves in time according to \eqref{eq:Hagedorn}, and one of the key ideas of these estimates is to find the Schr\"odinger-type equation satisfied by those wave packets and identify the residual term that accounts for the difference from the Schr\"odinger equation~\eqref{eq:Schroedinger}.

Following their approach, we would like to first find the Schr\"odinger-type equation satisfied by the time-dependent Gaussian wave packet~\eqref{eq:phi_0-t}.
The resulting residual term slightly differs from Hagedorn's because of the $O(\eps)$ correction term:
\begin{lemma}
  \label{lem:Schroedinger-phi_0}
  Consider the Gaussian $\phi_{0}(t,x)$ from \eqref{eq:phi_0-t} whose time-dependent parameters satisfy \eqref{eq:Symp_Hagedorn}.
  Then it satisfies the Schr\"odinger-type equation
  \begin{equation}
    \label{eq:Schroedinger-phi_0}
    \rmi\,\eps \pd{}{t}\phi_{0}(t,x) = \hat{H} \phi_{0}(t,x) + \eps^{3/2}\zeta_{0}(t,x),
  \end{equation}
  where we defined the residual term
  \begin{equation}
    \label{eq:zeta_0}
    \zeta_{0}(t,x) \defeq \alpha(q(t), Q(t); x)\,\phi_{0}(t,x)
  \end{equation}
  with
  \begin{equation}
    \label{eq:alpha}
    \alpha(q, Q; x)
    \defeq \eps^{-1/2}\partial_{q}V^{(1)}(q, Q) \cdot (x - q)
    + \eps^{-3/2} \parentheses{
      \sum_{k=0}^{2} \frac{1}{k!} D^{k}V(q) \cdot (x - q)^{k} - V(x)
    }.
  \end{equation}
  Furthermore, we may split $\alpha$ as
  \begin{subequations}
    \label{eq:alphas}
    \begin{equation}
      \alpha(q, Q; x) = \alpha^{(0)}(q, Q; x) + \eps^{1/2}\alpha^{(1)}(q; x)
    \end{equation}
    with
    \begin{align}
      \label{eq:alpha0}
      \alpha^{(0)}(q, Q; x)
      &\defeq \eps^{-1/2}\partial_{q}V^{(1)}(q, Q) \cdot (x - q) - \frac{\eps^{-3/2}}{3!} D^{3}V(q) \cdot (x - q)^{3} \nonumber\\
      &= \partial_{q}V^{(1)}(q, Q) \cdot \xi - \frac{1}{6} D^{3}V(q) \cdot \xi^{3},
      \\
      \alpha^{(1)}(q; x)
      &\defeq \eps^{-2} \parentheses{
        \sum_{k=0}^{3} \frac{1}{k!} D^{k}V(q) \cdot (x - q)^{k} - V(x)
        } \nonumber\\
       &= \frac{\eps^{-2}}{4!} D^{4}V(\sigma_{1}(x,q)) \cdot (x - q)^{4} \nonumber\\ 
      &= \frac{1}{4!} D^{4}V(\sigma_{1}(x,q)) \cdot \xi^{4}.
    \end{align}
  \end{subequations}
  Note that we set $\xi \defeq \eps^{-1/2}(\hat{x} - q)$ and $D^{0}V \defeq V$, and that $\sigma_{1}(x,q)$ is a point in the segment joining $x$ and $q$; we also used the shorthand $\xi^{m}$ with $m \in \N_{0}$ for the $m$-tensor defined as $\xi^{m}_{i_{1} \dots i_{m}} \defeq \xi_{i_{1}} \cdots \xi_{i_{m}}$ as well as
  \begin{equation*}
    D^{m}V(q) \cdot \xi^{m} \defeq D^{m}_{i_{1} \dots i_{m}}V(q) \xi^{m}_{i_{1} \dots i_{m}}
  \end{equation*}
  with Einstein's summation convention on repeated indices.
\end{lemma}

\begin{remark}
  Those terms with $D^{3}V$ in $\alpha^{(0)}$ and $\alpha^{(1)}$ cancel with each other in $\alpha$, but as we shall see below, splitting the terms in $\alpha$ in this manner is crucial for us as we shall see in the next subsection.
\end{remark}

\begin{proof}[Proof of Lemma~\ref{lem:Schroedinger-phi_0}]
  It follows from tedious but straightforward calculations:
  Dropping the time variable $t$ for brevity in the calculations, we have, using \eqref{eq:Symp_Hagedorn},
  \begin{align*}
    \eps^{3/2} \alpha(q(t), Q(t); x)
    &= \parentheses{ \DS \rmi\,\eps \pd{}{t}\phi_{0}(t,x) - \hat{H} \phi_{0}(t,x) } \bigg/ \phi_{0}(t,x) \\
    &= \parentheses{ \dot{q} - p }^{T} P Q^{-1} (x - q)
    - \frac{1}{2}(x - q)^{T}\parentheses{ \dot{P} Q^{-1} - P Q^{-1} \dot{Q} Q^{-1} + (P Q^{-1})^{2} } (x - q)
    \\
    &\quad - \dot{p} \cdot (x - q) + p \cdot \parentheses{ \dot{q} - p } - V(x) + V(q)
    \\
    &\quad + \parentheses{ \frac{p^{2}}{2} - V(q) - \dot{S} }
      - \frac{\rmi}{2} \eps \parentheses{ \tr(Q^{-1} \dot{Q})  - \tr(P Q^{-1}) }
    \\
    &= \eps\,\partial_{q}V^{(1)}(q,Q) \cdot (x - q) \\
    &\quad+ V(q) + DV(q) \cdot (x - q) + \frac{1}{2} D^{2}V(q) \cdot (x - q)^{2} - V(x),
  \end{align*}
  which gives \eqref{eq:alpha}.
  We may then split $\alpha$ as follows:
  \begin{align*}
    \alpha(q, Q; x)
    &= \eps^{-1/2}\,\partial_{q}V^{(1)}(q,Q) \cdot (x - q) - \frac{\eps^{-3/2}}{3!} D^{3}V(q) \cdot (x - q)^{3} \\
    &\quad + \eps^{-3/2}\parentheses{ \sum_{k=0}^{3} \frac{1}{k!} D^{k}V(q) \cdot (x - q)^{k} - V(x) } \\
    &= \alpha^{(0)}(q, Q; x) + \eps^{1/2}\alpha^{(1)}(q; x).
  \end{align*}
  The second expression for $\alpha^{(1)}$ follows from Taylor's Theorem because $V$ is of class $C^{4}$.
\end{proof}

\subsection{Properties of the Residual Term}
Let us prove some key properties of the residual term $\zeta_{0}$ as lemmas.
These lemmas show why we split $\alpha$ into $\alpha^{(0)}$ and $\alpha^{(1)}$ as shown in \eqref{eq:alphas}.
\begin{lemma}
  \label{lem:zeta_0-estimate}
  Under the assumptions on the potential $V$ from Theorem~\ref{thm:main}, we have the following estimates for the residual term $\zeta_{0}$ defined in \eqref{eq:zeta_0}:
  
  \begin{enumerate}[(i)]
  \item $\norm{ \alpha^{(0)}(q(t), Q(t);\,\cdot\,) \phi_{0}(t,\,\cdot\,) } = O(1)$;
    \label{lem:zeta_0-estimate1}
    \medskip
  \item $\norm{ \alpha^{(1)}(q(t);\,\cdot\,) \phi_{0}(t,\,\cdot\,) } = O(1)$;
    \label{lem:zeta_0-estimate2}
    \medskip
  \item $\norm{ \zeta_{0}(t,\,\cdot\,) } = O(1)$.
    \label{lem:zeta_0-estimate3}
    \medskip
  \item $\norm{ \hat{\xi}_{i}(t)\,\zeta_{0}(t,\,\cdot\,) } = O(1)$ for any $i \in \{1, \dots, d\}$ with $\hat{\xi}(t) \defeq \eps^{-1/2}(\hat{x} - q(t))$.
    \label{lem:zeta_0-estimate4}
    \medskip
  \item $\norm{ \hat{\eta}_{i}(t)\,\zeta_{0}(t,\,\cdot\,) } = O(1)$ for any $i \in \{1, \dots, d\}$ with $\hat{\eta}(t) \defeq \eps^{-1/2}(\hat{p} - p(t))$.
    \label{lem:zeta_0-estimate5}
  \end{enumerate}
\end{lemma}
\begin{proof}
  To prove~(\ref{lem:zeta_0-estimate1}), notice that it is the square root of the integral with respect to $x$ of
  \begin{equation*}
    |\phi_{0}(t,x)|^{2} = \frac{ |\det Q(t)|^{-1} }{ (\pi\eps)^{d/2} } \exp\parentheses{ -\frac{1}{\eps}(x - q(t))^{T}(Q(t) Q(t)^{*})^{-1}(x - q(t)) }
  \end{equation*}
  multiplied by a polynomial of $\eps^{-1/2}(x - q(t))$.
  It is straightforward to see that, by performing the integral using the change of variables from $x$ to $\xi \defeq \eps^{-1/2}(x - q(t))$, the integral does not depend on $\eps$.
  More specifically, for any $\mathbf{m} = (m_{1}, \dots, m_{d}) \in \N_{0}^{d}$,
  \begin{equation*}
    \norm{ \xi_{1}^{m_{1}} \dots \xi_{d}^{m_{d}} \, \phi_{0}(t,\,\cdot\,) } = O(1), 
  \end{equation*}
  and hence it follows that $\norm{ \alpha^{(0)}(q(t), Q(t);\,\cdot\,) \phi_{0}(t,\,\cdot\,) } = O(1)$; see, e.g., \citet[Eq.~(3.30)]{Ha1998} for an equivalent statement.
  
  For~(\ref{lem:zeta_0-estimate2}), we mimic the proof of Theorem~2.9 of \citet{Ha1998}.
  Take an arbitrarily small $r > 0$ and let $\bar{\mathbb{B}}_{r}(q(t)) \subset \R^{d}$ be the closed ball with radius $r$ centered at $q(t)$.
  Then, 
  \begin{equation*}
    \alpha^{(1)}(q(t); x) = \mathbf{1}_{\bar{\mathbb{B}}_{r}(q(t))}(x)\, \alpha^{(1)}(q(t); x) + \mathbf{1}_{\bar{\mathbb{B}}_{r}(q(t))^{\rm c}}(x)\, \alpha^{(1)}(q(t); x),
  \end{equation*}
  where $\mathbf{1}_{A}$ stands for the characteristic function for a subset $A \subset \R^{d}$.
  Since $V$ is of class $C^{4}$, there exists a 4-tensor-valued time-dependent function $F_{r}(t)$ such that, for any $x \in \bar{\mathbb{B}}_{r}(q(t))$,
  \begin{equation*}
    \abs{ \alpha^{(1)}(q(t); x) }
    = \abs{ \frac{1}{4!} D^{4}V(\sigma_{1}(x,q(t))) \cdot \xi^{4} }
    \le \abs{ F_{r}(t) \cdot \xi^{4} }.
  \end{equation*}
  On the other hand, there exist $C_{2}(t) > 0$ and a polynomial $\mathcal{P}(x)$ such that, for any $x \in \bar{\mathbb{B}}_{r}(q(t))^{\rm c}$, 
  \begin{equation*}
    \abs{ \alpha^{(1)}(q(t); x) } \le \eps^{-2} C_{2}(t) \mathcal{P}(x)
  \end{equation*}
  because the boundedness assumption on the Hessian $D^{2}V$ implies that the potential $V$ is dominated by a quadratic function on $\bar{\mathbb{B}}_{r}(q(t))^{\rm c}$, and also the rest of $\alpha^{(1)}(q(t); x)$ is cubic in $x$.
  Therefore,
  \begin{align*}
    \abs{ \alpha^{(1)}(q(t); x) \phi_{0}(t,x) }
    &= \mathbf{1}_{\bar{\mathbb{B}}_{r}(q(t))}(x)\,\abs{ \alpha^{(1)}(q(t); x) \phi_{0}(t,x) }
      + \mathbf{1}_{\bar{\mathbb{B}}_{r}(q(t))^{\rm c}}(x)\, \abs{ \alpha^{(1)}(q(t); x) \phi_{0}(t,x) } \\
    &\le \mathbf{1}_{\bar{\mathbb{B}}_{r}(q(t))}(x)\, \abs{ F_{r}(t) \cdot \xi^{4} \phi_{0}(t,x) }
     + \eps^{-2} C_{2}(t) \mathbf{1}_{\bar{\mathbb{B}}_{r}(q(t))^{\rm c}}(x)\, \mathcal{P}(x)\,\abs{ \phi_{0}(t,x) }.
  \end{align*}
  However, the norm of the first term is $O(1)$ following the same argument as in \eqref{lem:zeta_0-estimate1}, whereas the norm of the second term is $o(\eps^{r})$ for any real $r$ effectively canceling $\eps^{-2}$ in the coefficient.
  Hence $\norm{ \alpha^{(1)}(q(t);\,\cdot\,) \phi_{0}(t,\,\cdot\,) } = O(1)$.

  The estimate in (\ref{lem:zeta_0-estimate3}) follows easily from (\ref{lem:zeta_0-estimate1}) and (\ref{lem:zeta_0-estimate2}):
  \begin{equation*}
     \norm{ \zeta_{0}(t,\,\cdot\,) }
     \le \bigl\| \alpha^{(0)}(q(t), Q(t);\,\cdot\,) \phi_{0}(t,\,\cdot\,) \bigr\|
     + \eps^{1/2} \bigl\| \alpha^{(1)}(q(t);\,\cdot\,) \phi_{0}(t,\,\cdot\,) \bigr\| = O(1).
   \end{equation*}

   The estimate in (\ref{lem:zeta_0-estimate4}) holds similarly because the above estimates do not change upon multiplying $\phi_{0}(t,x)$ by $\hat{\xi}_{i}(t) \defeq \eps^{-1/2}(\hat{x} - q(t))_{i}$.
   
   The estimate in (\ref{lem:zeta_0-estimate5}) holds because straightforward calculations (see Appendix~\ref{sseca:zeta_0-estimate5}) show that
   \begin{equation}
     \label{eq:eta_zeta}
     \hat{\eta}_{i}(t) \zeta_{0}(t,x) = \beta_{i}(q(t),Q(t);x) \phi_{0}(t,x) + (P(t) Q(t)^{-1})_{ij} \xi_{j}(t) \zeta_{0}(t,x),
   \end{equation}
   where $\beta_{i}(q,Q; x) \defeq \beta_{i}^{(0)}(q,Q;x) + \eps^{1/2}\beta_{i}^{(1)}(q;x)$ with
   \begin{equation}
     \label{eq:betas}
     \begin{split}
       \beta_{i}^{(0)}(q,Q;x) &\defeq -\rmi \parentheses{
         \partial_{q_{i}}V^{(1)}(q,Q) - \frac{1}{2} D^{3}_{ijk}V(q) \xi^{2}_{jk}
       }, \\
       \beta_{i}^{(1)}(q;x) &\defeq - \rmi\,\eps^{-3/2} \parentheses{
         D_{i}V(q) + D^{2}_{ij}V(q) (x - q)_{j} + \frac{1}{2} D^{3}_{ijk}V(q) (x - q)^{2}_{jk} - D_{i}V(x)
       }.
     \end{split}
   \end{equation}
   For the first term on the right-hand side of \eqref{eq:eta_zeta}, notice the similarity between $\beta_{i}$ and $\alpha$; so we can obtain the estimate of the first term essentially the same way we did for $\zeta_{0}$, i.e., its norm is $O(1)$.
   We also know from (\ref{lem:zeta_0-estimate4}) that the norm of the second term in \eqref{eq:eta_zeta} is $O(1)$ as well.
\end{proof}

Furthermore, the first part $\alpha^{(0)} \phi_{0}$ of the residual term $\zeta_{0}$ satisfies the following orthogonality property that later turns out to be crucial:
\begin{lemma}
  \label{lem:orthogonality}
  For any multi-index $\mathbf{k} \in \N_{0}^{d}$ with $0 \le |\mathbf{k}| \le 2$, we have
  \begin{equation*}
    \ip{ \alpha^{(0)}(q, Q; \,\cdot\,)\, \phi_{0}(q, p, Q, P, S; \,\cdot\,) }{ \phi_{\mathbf{k}}(q, p, Q, P, S; \,\cdot\,) } = 0.
  \end{equation*}
  More specifically, $\alpha^{(0)}(q, Q; \,\cdot\,)\, \phi_{0}(q, p, Q, P, S; \,\cdot\,)$ is a linear combination of the third excited states
  \begin{equation*}
    \setdef{ \phi_{\mathbf{n}}(q, p, Q, P, S; \,\cdot\,) }{ \mathbf{n} \in \N_{0}^{d} \text{ with } |\mathbf{n}| = 3}.
  \end{equation*}
\end{lemma}
\begin{proof}
  Substituting \eqref{eq:V^1} into the expression~\eqref{eq:alpha0} for $\alpha^{(0)}$, we have, suppressing the variables in $\alpha^{(0)}$ and $\phi_{0}$ for brevity,
  \begin{equation*}
    \alpha^{(0)} \phi_{0}
    = \parentheses{
      \frac{\eps^{-1/2}}{4} Q_{jl} \overline{Q}_{kl} D^{3}_{ikj}V(q) (x - q)_{i}
      - \frac{\eps^{-3/2}}{3!} D^{3}_{ijk}V(q) (x - q)^{3}_{ijk}
    } \phi_{0}.
  \end{equation*}
  Using \eqref{eq:hatx-q_in_ladder_ops} and noting \eqref{eq:A-phi_0}, the first term becomes
  \begin{align}
    \label{eq:alpha0-1st_term}
    \frac{\eps^{-1/2}}{4} Q_{jl} \overline{Q}_{kl} D^{3}_{ikj}V(q) (x - q)_{i} \phi_{0}
    &= \frac{1}{4\sqrt{2}} D^{3}_{ikj}V(q) Q_{jl} \overline{Q}_{kl} Q_{in} \mathscr{A}^{*}_{n} \phi_{0} \nonumber\\
    &= \frac{1}{4\sqrt{2}} \mathcal{T}_{ijk} \overline{Q}_{il} Q_{jl} Q_{kn} \phi_{\mathbf{e}_{n}},
  \end{align}
  where we used the shorthand $\mathcal{T}_{ijk} \defeq D^{3}_{ijk}V(q)$ and its symmetry with respect to permutations of the indices as well as \eqref{eq:phi_n-raised}.
  On the other hand, after similar but more tedious calculations (see Appendix~\ref{sseca:orthogonality}), the second term becomes
  \begin{equation}
    \label{eq:alpha0-2nd_term}
    -\frac{\eps^{-3/2}}{3!} D^{3}_{ijk}V(q) (x - q)^{3}_{ijk} \phi_{0}
    = -\frac{1}{4\sqrt{3}} \mathcal{T}_{ijk} Q_{il} Q_{jm} Q_{kn} \phi_{\mathbf{e}_{l} + \mathbf{e}_{m} + \mathbf{e}_{n}}
    - \frac{1}{4\sqrt{2}} \mathcal{T}_{ijk} \overline{Q}_{il} Q_{jl} Q_{kn} \phi_{\mathbf{e}_{n}}.
  \end{equation}
  As a result,
  \begin{equation*}
    \alpha^{(0)} \phi_{0}
    = -\frac{1}{4\sqrt{3}} \mathcal{T}_{ijk} Q_{il} Q_{jm} Q_{kn} \phi_{\mathbf{e}_{l} + \mathbf{e}_{m} + \mathbf{e}_{n}},
  \end{equation*}
  which is a linear combination of the third excited states.
\end{proof}

Notice the role played by the first term in $\alpha^{(0)}$ in canceling the term with the first excited states, and recall that this term~\eqref{eq:alpha0-1st_term} came from the $O(\eps)$ correction term in the potential in our dynamics~\eqref{eq:Symp_Hagedorn-qpQP}.

\begin{remark}
  \label{rem:Hagedorn1}
  What if one uses the classical Hamiltonian system for $t \mapsto (q(t),p(t))$ as in \eqref{eq:Hagedorn} of Hagedorn?
  Then the function $\alpha$ in the residual term $\zeta_{0}$ becomes
  \begin{align}
    \label{eq:alpha-Hagedorn}
    \alpha(q,Q;x)
    &= \eps^{-3/2} \parentheses{ \sum_{k=0}^{2} \frac{1}{k!} D^{k}V(q) \cdot (x - q)^{k} - V(x) } \nonumber\\
    &= -\eps^{-3/2} \parentheses{ \frac{1}{3!} D^{3}V(q) \cdot (x - q)^{3} + \frac{1}{4!} D^{4}V(\sigma_{1}(x,q)) \cdot (x - q)^{4} },
  \end{align}
  where $\sigma_{1}(x,q)$ is defined in Lemma~\ref{lem:Schroedinger-phi_0}.
  The absence of the term coming from the correction term indicates that there is no cancellation of those terms involving the first excited states.
  Indeed, as we shall see in Remark~\ref{rem:Hagedorn2} of Section~\ref{ssec:estimate1}, this residual term does not enjoy the same property as ours does; it turns out to be detrimental in the error estimate.
\end{remark}

\section{Time Evolution of the Hagedorn Wave Packets}
\label{sec:time_evolution-HWP}
\subsection{Overview}
While the main focus of the paper is the Gaussian~\eqref{eq:phi_0} and its associated equations~\eqref{eq:Symp_Hagedorn} for the parameters, it turns out that the proof of the main result requires some analysis on the time evolution of some other Hagedorn wave packets as well.
Therefore, in this section, we derive the Schr\"odinger-type evolution equation for the Hagedorn wave packets as opposed to just the Gaussian.

\subsection{Evolution Equation of the Hagedorn Wave Packets}
Lemma~\ref{lem:Schroedinger-phi_0} applies only to the Gaussian wave packet $\phi_{0}$.
However, it turns out that Lemma~\ref{lem:Schroedinger-phi_0} generalizes to $\phi_{\mathbf{n}}$ with any $\mathbf{n} \in \N_{0}^{d}$:
\begin{proposition}
  \label{prop:Schroedinger-phi_n}
  Let us define, for $\mathbf{n} \in \N_{0}^{d}$, 
  \begin{align*}
    \phi_{\mathbf{n}}(t,x) &\defeq \phi_{\mathbf{n}}(q(t), p(t), Q(t), P(t), S(t); x),
  \end{align*}
  where $t \mapsto (q(t), p(t), Q(t), P(t), S(t))$ satisfies \eqref{eq:Symp_Hagedorn}.
  Then $\phi_{\mathbf{n}}(t, x)$ satisfies the Schr\"odinger-type equation
  \begin{equation}
    \label{eq:Schroedinger-phi_n}
    \rmi\,\eps \pd{}{t}\phi_{\mathbf{n}}(t,x) = \hat{H} \phi_{\mathbf{n}}(t,x) + \eps^{3/2}\zeta_{\mathbf{n}}(t,x),
  \end{equation}
  where
  \begin{equation}
    \label{eq:zeta_n}
    \zeta_{\mathbf{n}}(t,x) \defeq \alpha(q(t), Q(t); x) \phi_{\mathbf{n}}(t,x)
  \end{equation}
  with the same $\alpha$ defined in \eqref{eq:alpha}.
\end{proposition}

This result follows easily from the following lemma regarding the time evolution of the raising operator:
\begin{lemma}
  \label{lem:time_evolution_of_raising_op}
  Suppose that $t \mapsto (q(t), p(t), Q(t), P(t), S(t))$ satisfies \eqref{eq:Symp_Hagedorn} and let us write
  \begin{equation*}
    \mathscr{A}^{*}(t) \defeq \mathscr{A}^{*}(q(t), p(t), Q(t), P(t)).
  \end{equation*}
  Then its time evolution is governed by
  \begin{equation*}
    \rmi\,\eps \od{}{t}\mathscr{A}^{*}(t) + \bigl[ \mathscr{A}^{*}(t), \hat{H} \bigr]
    = \frac{\eps^{2}}{\sqrt{2}}\, Q^{*}(t) \partial_{x} \alpha(q(t), Q(t); x)
  \end{equation*}
  as operators with domain $\mathscr{S}(\R^{d})$.
\end{lemma}
\begin{proof}
  It is straightforward to see, in view of \eqref{eq:Symp_Hagedorn} and \eqref{eq:Astar}, that
  \begin{align*}
    \rmi\,\eps \od{}{t}\mathscr{A}^{*}(t)
    &= \sqrt{\frac{\eps}{2}} \parentheses{
      -\dot{P}^{*}(t) (x - q(t)) + P^{*}(t)\dot{q}(t) + \dot{Q}^{*}(t) (\hat{p} - p(t)) - Q^{*}(t) \dot{p}(t)
    } \\
    &= \sqrt{\frac{\eps}{2}} \parentheses{
      Q^{*}(t) D^{2}V(q(t)) (x - q(t)) + P^{*}(t)\hat{p} + Q^{*}(t)\parentheses{ DV(q(t)) + \eps\,\partial_{q} V^{(1)}(q(t), Q(t)) }
    },
  \end{align*}
  whereas, for $j \in \{1, \dots, d\}$, 
  \begin{align*}
    \bigl[ \mathscr{A}^{*}_{j}(q, p, Q, P), \hat{H} \bigr]
    &= \frac{\rmi}{\sqrt{2\eps}}\parentheses{
      \frac{1}{2}P^{*}_{jk}\brackets{ x_{k}, \hat{p}^{2} }
      - Q^{*}_{jk} \brackets{ \hat{p}_{k}, V(x) }
      } \\
    &= -\sqrt{\frac{\eps}{2}} \parentheses{
      P^{*}_{jk}\hat{p}_{k}
      + Q^{*}_{jk} \pd{V}{x_{k}}(x)
      }
  \end{align*}
  and so
  \begin{equation*}
    \bigl[ \mathscr{A}^{*}(t), \hat{H} \bigr]
    = -\sqrt{\frac{\eps}{2}} \parentheses{
      P^{*}(t) \hat{p}
      + Q^{*}(t) D V(x)
      }.
  \end{equation*}
  Therefore,
  \begin{align*}
    \rmi\,\eps& \od{}{t}\mathscr{A}^{*}(t) + \bigl[ \mathscr{A}^{*}(t), \hat{H} \bigr] \\
    &= \sqrt{\frac{\eps}{2}}\, Q^{*}(t) \parentheses{
      D^{2}V(q(t)) (x - q(t)) + DV(q(t)) + \eps\,\partial_{q}V^{(1)}(q(t), Q(t))
      - DV(x)
      } \\
    &= \sqrt{\frac{\eps}{2}}\, Q^{*}(t) \pd{}{x} \Bigl(
      \eps\,\partial_{q} V^{(1)}(q(t), Q(t)) \cdot (x - q(t))
    \\
    &\phantom{ =\sqrt{\frac{\eps}{2}} Q^{*}(t) \pd{}{x} \Bigl(\ }
      + V(q(t))
      + DV(q(t)) \cdot (x - q(t))
      + \frac{1}{2}D^{2}V(q(t)) \cdot (x - q(t))^{2}
      - V(x)
      \Bigr) \\
    &= \frac{\eps^{2}}{\sqrt{2}}\, Q^{*}(t) \partial_{x} \alpha(q(t), Q(t); x). \qedhere
  \end{align*}
\end{proof}

\begin{proof}[Proof of Proposition~\ref{prop:Schroedinger-phi_n}]
  By induction on $\mathbf{n} \in \N_{0}^{d}$.
  Lemma~\ref{lem:Schroedinger-phi_0} shows that the assertion holds for $\mathbf{n} = 0$.
  Let us suppose that the assertion holds for $\mathbf{n} \in \N_{0}^{d}$ and show that it holds for $\mathbf{n} + \mathbf{e}_{j}$ for any $j \in \{1, \dots, d\}$.
  Using \eqref{eq:phi_n-raised} and the above lemma,
  \begin{align*}
    \sqrt{n_{j} + 1}& \parentheses{
                         \rmi\,\eps \pd{}{t}\phi_{\mathbf{n}+\mathbf{e}_{j}}(t, x) - \hat{H} \phi_{\mathbf{n}+\mathbf{e}_{j}}(t, x)
                         }\\
                       &= \rmi\,\eps \pd{}{t}\parentheses{ \mathscr{A}^{*}_{j}(t)\, \phi_{\mathbf{n}}(t, x) } - \hat{H} \mathscr{A}^{*}_{j}(t)\, \phi_{\mathbf{n}}(t, x) \\
                       &= \parentheses{ \rmi\,\eps \od{}{t}\mathscr{A}^{*}_{j}(t) + \bigl[ \mathscr{A}^{*}_{j}(t), \hat{H} \bigr] } \phi_{\mathbf{n}}(t, x)
                         + \mathscr{A}^{*}_{j}(t) \parentheses{ \rmi\,\eps \pd{}{t}\phi_{\mathbf{n}}(t, x) - \hat{H} \phi_{\mathbf{n}}(t, x) } \\
                       &= \frac{\eps^{2}}{\sqrt{2}} Q^{*}_{jk}(t) \pd{}{x_{k}}\alpha(q(t), Q(t); x) \phi_{\mathbf{n}}(t, x)
                         + \eps^{3/2}\mathscr{A}^{*}_{j}(t) \parentheses{ \alpha(q(t), Q(t); x) \phi_{\mathbf{n}}(t, x) }.
  \end{align*}
  However, since every term in $\mathscr{A}^{*}$ (see \eqref{eq:Astar}) except the one with $\hat{p}$ is a multiplication operator,
  \begin{align*}
    \mathscr{A}^{*}(q, p, Q, P)& \parentheses{ \alpha(q, Q; x) \phi_{\mathbf{n}}(q, p, Q, P, S; x) } \\
    &= \alpha(q, Q; x) \mathscr{A}^{*}(q, p, Q, P) \phi_{\mathbf{n}}(q, p, Q, P, S; x) \\
      &\quad - \frac{\rmi}{\sqrt{2\eps}} Q^{*} \hat{p}\parentheses{ \alpha(q, Q; x) } \phi_{\mathbf{n}}(q, p, Q, P, S; x) \\
    &= \sqrt{n_{j} + 1}\,\alpha(q, Q; x) \phi_{\mathbf{n}+\mathbf{e}_{j}}(q, p, Q, P, S; x)
      - \sqrt{\frac{\eps}{2}} Q^{*} \partial_{x} \alpha(q, Q; x) \phi_{\mathbf{n}}(q, p, Q, P, S; x).
  \end{align*}
  Therefore, we obtain
  \begin{equation*}
    \rmi\,\eps \pd{}{t}\phi_{\mathbf{n}+\mathbf{e}_{j}}(t, x) - \hat{H} \phi_{\mathbf{n}+\mathbf{e}_{j}}(t, x)
    = \eps^{3/2} \alpha(q(t), Q(t); x) \phi_{\mathbf{n}+\mathbf{e}_{j}}(t, x). \qedhere
  \end{equation*}
\end{proof}

\subsection{Errors in Wave Functions}
Let us first note that, in what follows, we will suppress the spatial variables $x$ for brevity.
We also note that the assumption that the potential $V$ is bounded from below guarantees that there exists a self-adjoint extension of the Schr\"odinger operator $\hat{H}$ so that the unitary operators of the form $e^{-\rmi\hat{H}t/\eps}$ with $t \in \R$ would make sense.

Now we would like to compare the exact solution $t \mapsto e^{-\rmi\hat{H}(t-s)/\eps} \phi_{\mathbf{n}}(s)$ of the Schr\"odinger equation~\eqref{eq:Schroedinger} and the wave packet $t \mapsto \phi_{\mathbf{n}}(t)$, both with the initial wave function being $\phi_{\mathbf{n}}(s)$ with any $\mathbf{n} \in \N_{0}^{d}$ at time $s \in \R$.
To that end, let us define the difference between them (i.e., error in wave functions): For any $\mathbf{n} \in \N_{0}^{d}$, and any $s, t \in \R$ (for which both $\phi_{\mathbf{n}}(s)$ and $\phi_{\mathbf{n}}(t)$ are defined),
\begin{equation}
  \label{eq:mathcalZ}
  \mathcal{Z}_{\mathbf{n}}(t,s) \defeq e^{-\rmi\hat{H}(t-s)/\eps} \phi_{\mathbf{n}}(s) - \phi_{\mathbf{n}}(t),
  \qquad
  \mathcal{Z}_{\mathbf{n}}(t) \defeq \mathcal{Z}_{\mathbf{n}}(t,0) = e^{-\rmi\hat{H}t/\eps} \phi_{\mathbf{n}}(0) - \phi_{\mathbf{n}}(t).
\end{equation}

The following lemma is critical in finding an estimate of these errors:
\begin{lemma}
  \label{lem:zeta_ell-estimate}
  For any $\mathbf{n} \in \N_{0}^{d}$, $\norm{ \zeta_{\mathbf{n}}(t) } = O(1)$.
\end{lemma}

\begin{proof}
  The proof is almost identical to that of Lemma~\ref{lem:zeta_0-estimate}.
  In fact, one can show that, for any $\mathbf{m}, \mathbf{n} \in \N_{0}^{d}$,
  \begin{equation*}
    \norm{ \xi_{1}^{m_{1}} \dots \xi_{d}^{m_{d}} \, \phi_{\mathbf{n}}(t) } = O(1)
  \end{equation*}
  because $\phi_{\mathbf{n}}$ is $\phi_{0}$ multiplied by an $|\mathbf{n}|$-th order polynomial of $\xi = \eps^{-1/2}(x - q(t))$; see also \citet[Eq.~(3.30)]{Ha1998}.
  It implies that those arguments with $\phi_{0}$ from Lemma~\ref{lem:zeta_0-estimate} still apply upon replacing $\phi_{0}$ by $\phi_{\mathbf{n}}$.
  Hence it follows that $\norm{ \alpha^{(0)}(q(t), Q(t)) \phi_{\mathbf{n}}(t) } = O(1)$ as well as that $\norm{ \alpha^{(1)}(q(t)) \phi_{\mathbf{n}}(t) } = O(1)$ as well.
\end{proof}

As a result, we have an expression and an estimate for $\mathcal{Z}_{\mathbf{n}}$ as follows:
\begin{proposition}
  \label{prop:magic_formula}
  The errors defined in \eqref{eq:mathcalZ} can be written in terms of the residual term $\zeta_{\mathbf{n}}$ from \eqref{eq:zeta_n} as follows:
  For any $\mathbf{n} \in \N_{0}^{d}$ and any $s, t \in \R$ for which $\phi_{\mathbf{n}}(s)$ and $\phi_{\mathbf{n}}(t)$ are defined,
  \begin{equation}
    \label{eq:magic_formula}
    \mathcal{Z}_{\mathbf{n}}(t,s) \defeq e^{-\rmi\hat{H}(t-s)/\eps} \phi_{\mathbf{n}}(s) - \phi_{\mathbf{n}}(t)
    = \rmi\,\eps^{1/2} \int_{s}^{t} e^{-\rmi\hat{H}(t-\tau)/\eps} \zeta_{\mathbf{n}}(\tau)\,d\tau,
  \end{equation}
  and hence $\norm{ \mathcal{Z}_{\mathbf{n}}(t,s) } = O(\eps^{1/2})$ in the sense that there exists some function $\mathscr{C}$ such that $\norm{ \mathcal{Z}_{\mathbf{n}}(t,s) } \le \mathscr{C}(t,s) \eps^{1/2}$.
\end{proposition}
\begin{proof}
  This is essentially the same as the proof of \citet[Lemma~2.8]{Ha1998}, but we briefly reproduce it here for completeness.
  Using the Schr\"odinger-type equation~\eqref{eq:Schroedinger-phi_n} satisfied by $\tau \mapsto \phi_{\mathbf{n}}(\tau)$, we have
  \begin{equation*}
    \pd{}{\tau}\parentheses{ e^{-\rmi\hat{H}(t-\tau)/\eps} \phi_{\mathbf{n}}(\tau) } = -\rmi\,\eps^{1/2} e^{-\rmi\hat{H}(t-\tau)/\eps}\zeta_{\mathbf{n}}(\tau).
  \end{equation*}
  Integrating both sides with respect to $\tau$ over the time interval between $s$ and $t$ yields
  \begin{equation*}
    \phi_{\mathbf{n}}(t) - e^{-\rmi\hat{H}(t-s)/\eps} \phi_{\mathbf{n}}(s) = -\rmi\,\eps^{1/2} \int_{s}^{t} e^{-\rmi\hat{H}(t-\tau)/\eps}\zeta_{\mathbf{n}}(\tau)\,d\tau.
  \end{equation*}
  The left hand side is $-\mathcal{Z}_{\mathbf{n}}(t,s)$, and so \eqref{eq:magic_formula} follows.
  The estimate in norm follows by taking the norm of both sides of \eqref{eq:magic_formula}:
  \begin{equation*}
    \norm{ \mathcal{Z}_{\mathbf{n}}(t,s) } \le \eps^{1/2} \int_{s}^{t} \norm{\zeta_{\mathbf{n}}(\tau) }\,d\tau = O(\eps^{1/2}).
  \end{equation*}
  due to the unitarity of $e^{-\rmi\hat{H}(t-s)/\eps}$ as well as Lemma~\ref{lem:zeta_ell-estimate}.
\end{proof}

Particularly, setting $s = 0$, we have the following:
\begin{corollary}
  \label{cor:magic_formula}
  Let $t \mapsto \psi(t)$ be the solution to the Schr\"odinger equation~\eqref{eq:Schroedinger} with the initial condition $\psi(0) = \phi_{\mathbf{n}}(0)$ with $\mathbf{n} \in \N_{0}^{d}$.
  Then
  \begin{equation}
    \label{eq:magic_formula0}
    \mathcal{Z}_{\mathbf{n}}(t) = \psi(t) - \phi_{\mathbf{n}}(t) = \rmi\,\eps^{1/2} \int_{0}^{t} e^{-\rmi\hat{H}(t-s)/\eps} \zeta_{\mathbf{n}}(s)\,ds,
  \end{equation}
  and hence $\norm{ \psi(t) - \phi_{\mathbf{n}}(t) } = O(\eps^{1/2})$.
\end{corollary}

The above result reproduces those estimates obtained by \citet{Ha1980,Ha1998} using our equations~\eqref{eq:Symp_Hagedorn}, and also indicates that using \eqref{eq:Symp_Hagedorn} in place of Hagedorn's \eqref{eq:Hagedorn} does not improve the errors in wave function in terms of $L^{2}$-norm---at least not with the above method of estimation.
The reason why there is still a difference in the error estimates of the observables as stated in Theorem~\ref{thm:main} is that our estimates involve a more detailed analysis of the residual term $\zeta_{0}$ as opposed to just having an $L^{2}$-norm estimate of it.

\section{Proof of Main Result}
\label{sec:proof}
\subsection{Error Terms in Observables}
Let $t \mapsto \psi(t)$ be the exact solution of the initial value problem~\eqref{ivp:Schroedinger} of the Schr\"odinger equation.
From the definition of $\mathcal{Z}_{0}$ in \eqref{eq:mathcalZ} with $\mathbf{n} = 0$, we have $\psi(t) = \phi_{0}(t) + \mathcal{Z}_{0}(t)$, and so, as we have shown in Section~\ref{ssec:outline}, 
\begin{align*}
  \exval{\hat{x}}(t) - q(t)
  = 2\Re\ip{\mathcal{Z}_{0}(t)}{(\hat{x} - q(t)) \phi_{0}(t)} + \ip{\mathcal{Z}_{0}(t)}{(\hat{x} - q(t)) \mathcal{Z}_{0}(t)},
\end{align*}
and similarly,
\begin{equation*}
  \exval{\hat{p}}(t) - p(t)
  = 2\Re\ip{\mathcal{Z}_{0}(t)}{(\hat{p} - p(t)) \phi_{0}(t)} + \ip{\mathcal{Z}_{0}(t)}{(\hat{p} - p(t)) \mathcal{Z}_{0}(t)}.
\end{equation*}

In the remaining subsections, we finish the proof of Theorem~\ref{thm:main} by showing that the two terms on the right-hand side of each of the above equations are both $O(\eps^{3/2})$.

\subsection{Estimates for First Error Term}
\label{ssec:estimate1}
First we see that, using the expression \eqref{eq:magic_formula0} for $\mathcal{Z}_{0}$ and Fubini's Theorem,
\begin{equation}
  \label{eq:ip1}
  \ip{\mathcal{Z}_{0}(t)}{(\hat{x} - q(t)) \phi_{0}(t)}
  = -\rmi\,\eps^{1/2} \int_{0}^{t} \ip{ e^{-\rmi\hat{H}(t-s)/\eps}\zeta_{0}(s) }{ (\hat{x} - q(t)) \phi_{0}(t) }\,ds.
\end{equation}
However, we can rewrite the inner product inside the integral as follows using the relationship \eqref{eq:hatx-q_in_ladder_ops} between the operator $\hat{x} - q$ and the ladder operators: For any $i \in \{1, \dots, d\}$,
\begin{align*}
  \ip{ e^{-\rmi\hat{H}(t-s)/\eps}\zeta_{0}(s) }{ (\hat{x} - q(t))_{i}\,\phi_{0}(t) }
  &= \sqrt{ \frac{\eps}{2} } \ip{ e^{-\rmi\hat{H}(t-s)/\eps}\zeta_{0}(s) }{ \parentheses{ \overline{Q}_{ij}(t) \mathscr{A}_{j}(t) + Q_{ij}(t) \mathscr{A}_{j}^{*}(t) } \phi_{0}(t) } \\
  &= \sqrt{ \frac{\eps}{2} }\,Q_{ij}(t)\ip{ e^{-\rmi\hat{H}(t-s)/\eps}\zeta_{0}(s) }{ \phi_{\mathbf{e}_{j}}(t) } \\
  &= \sqrt{ \frac{\eps}{2} }\,Q_{ij}(t)\ip{ \zeta_{0}(s) }{ e^{-\rmi\hat{H}(s-t)/\eps}\phi_{\mathbf{e}_{j}}(t) } \\
  &= \sqrt{ \frac{\eps}{2} }\,Q_{ij}(t)\parentheses{
    \ip{ \zeta_{0}(s) }{ \phi_{\mathbf{e}_{j}}(s) }
    + \ip{ \zeta_{0}(s) }{ \mathcal{Z}_{\mathbf{e}_{j}}(s,t) }
    }
\end{align*}
where $\mathscr{A}^{*}(t)$ is defined in Lemma~\ref{lem:time_evolution_of_raising_op} and similarly for $\mathscr{A}(t)$; we used \eqref{eq:phi_n-raised} for the second equality, and \eqref{eq:magic_formula} ($s$ and $t$ swapped) with $\mathbf{n} = \mathbf{e}_{j}$ for the last equality, i.e., for any $j \in \{1, \dots, d\}$, $e^{-\rmi\hat{H}(s-t)/\eps}\phi_{\mathbf{e}_{j}}(t) = \phi_{\mathbf{e}_{j}}(s) + \mathcal{Z}_{\mathbf{e}_{j}}(s,t)$.

Let us evaluate the above two terms:
First, recalling the formulas~\eqref{eq:zeta_0} and \eqref{eq:alphas} for $\zeta_{0}$ and exploiting the orthogonality in Lemma~\ref{lem:orthogonality}, we have, for any $j \in \{1, \dots, d\}$,
\begin{align*}
  \ip{ \zeta_{0}(s) }{ \phi_{\mathbf{e}_{j}}(s) }
  &= \ip{ \alpha^{(0)}(q(s), Q(s)) \phi_{0}(s) }{ \phi_{\mathbf{e}_{j}}(s) }
    + \eps^{1/2}\ip{ \alpha^{(1)}(q(s)) \phi_{0}(s) }{ \phi_{\mathbf{e}_{j}}(s) } \\
  &= \eps^{1/2}\ip{ \alpha^{(1)}(q(s)) \phi_{0}(s) }{ \phi_{\mathbf{e}_{j}}(s) },
\end{align*}
and thus by the Cauchy--Schwarz inequality and Lemma~\ref{lem:zeta_0-estimate}~(\ref{lem:zeta_0-estimate2}),
\begin{gather*}
  \abs{ \ip{ \zeta_{0}(s) }{ \phi_{\mathbf{e}_{j}}(s) } } \le \eps^{1/2}\bigl\| \alpha^{(1)}(q(s)) \phi_{0}(s) \bigr\| \norm{ \phi_{\mathbf{e}_{j}}(s) } = O(\eps^{1/2}).
\end{gather*}
On the other hand, again by the Cauchy--Schwarz inequality, Lemma~\ref{lem:zeta_0-estimate}~(\ref{lem:zeta_0-estimate3}), and Proposition~\ref{prop:magic_formula}, we have, for any $j \in \{1, \dots, d\}$,
\begin{equation*}
  \abs{ \ip{ \zeta_{0}(s) }{ \mathcal{Z}_{\mathbf{e}_{j}}(s,t) } }
  \le \norm{ \zeta_{0}(s) } \norm{ \mathcal{Z}_{\mathbf{e}_{j}}(s,t) } 
  = O(\eps^{1/2}).
\end{equation*}
Hence we see that
\begin{equation*}
  \abs{ \ip{ e^{-\rmi\hat{H}(t-s)/\eps}\zeta_{0}(s) }{ (\hat{x} - q(t))_{i} \phi_{0}(t) } } = O(\eps),
\end{equation*}
and therefore \eqref{eq:ip1} yields, for any $i \in \{1, \dots, d\}$,
\begin{equation*}
  \abs{ \ip{\mathcal{Z}_{0}(t)}{(\hat{x} - q(t))_{i} \phi_{0}(t)} }
  \le \eps^{1/2} \int_{0}^{t} \abs{ \ip{ e^{-\rmi\hat{H}(t-s)/\eps}\zeta_{0}(s) }{ (\hat{x} - q(t)) \phi_{0}(t) } } \,ds
  = O(\eps^{3/2}).
\end{equation*}

Using the relationship \eqref{eq:hatp-p_in_ladder_ops} between the operator $\hat{p} - p$ and the ladder operators, we can proceed in the same way to obtain
\begin{equation*}
  \abs{ \ip{\mathcal{Z}_{0}(t)}{(\hat{p} - p(t))_{i} \phi_{0}(t)} } = O(\eps^{3/2})
\end{equation*}
for any $i \in \{1, \dots, d\}$ as well.

\begin{remark}
  \label{rem:Hagedorn2}
  What if one uses the classical Hamiltonian system for $(q,p)$ as in \eqref{eq:Hagedorn} of Hagedorn?
  As discussed in Remark~\ref{rem:Hagedorn1}, we have $\alpha$ as shown in \eqref{eq:alpha-Hagedorn}.
  Then, as shown in Appendix~\ref{sseca:Hagedorn2}, we have
  \begin{equation*}
    \ip{ \zeta_{0}(s) }{ \phi_{\mathbf{e}_{j}}(s) } = O(1)
  \end{equation*}
  in this case as opposed to $O(\eps^{1/2})$.
  Indeed the leading $O(1)$ term (see \eqref{eq:ip-Hagedorn2} in Appendix~\ref{sseca:Hagedorn2}) is exactly the term canceled due to the first term in \eqref{eq:alpha} coming from the correction term in our case.
  This underscores the importance of the correction term alluded in Remark~\ref{rem:Hagedorn1}.
  As a result, the above estimates become
  \begin{equation*}
    \abs{ \ip{\mathcal{Z}_{0}(t)}{(\hat{x} - q(t))_{i} \phi_{0}(t)} } = O(\eps),
    \qquad
    \abs{ \ip{\mathcal{Z}_{0}(t)}{(\hat{p} - p(t))_{i} \phi_{0}(t)} } = O(\eps).
  \end{equation*}
  as opposed to $O(\eps^{3/2})$.
\end{remark}

\subsection{Estimates for Second Error Term}
It now remains to show that, for any $i \in \{1, \dots, d\}$,
\begin{equation*}
  \abs{ \ip{\mathcal{Z}_{0}(t)}{(\hat{x} - q(t))_{i} \mathcal{Z}_{0}(t)} } = O(\eps^{3/2}),
  \qquad
  \abs{ \ip{\mathcal{Z}_{0}(t)}{(\hat{p} - p(t))_{i} \mathcal{Z}_{0}(t)} } = O(\eps^{3/2}).
\end{equation*}

From \eqref{eq:Schroedinger} and \eqref{eq:Schroedinger-phi_0}, we see that
\begin{equation*}
  \dot{\mathcal{Z}}_{0}(t)
  = \dot{\psi}(t) - \dot{\phi}_{0}(t) = 
  -\frac{\rmi}{\eps} \hat{H} \mathcal{Z}_{0}(t)
  + \rmi\,\eps^{1/2} \zeta_{0}(t).
\end{equation*}
Therefore,
\begin{align*}
  \od{}{t} \parentheses{ (\hat{x} - q(t)) \mathcal{Z}_{0}(t) }
  &= -\dot{q}(t)\, \mathcal{Z}_{0}(t) + (\hat{x} - q(t)) \dot{\mathcal{Z}}_{0}(t) \\
  &= -p(t)\, \mathcal{Z}_{0}(t) - \frac{\rmi}{\eps} (\hat{x} - q(t)) \hat{H} \mathcal{Z}_{0}(t)
    + \rmi\,\eps^{1/2}  (\hat{x} - q(t)) \zeta_{0}(t) \\
  &= -p(t)\, \mathcal{Z}_{0}(t) - \frac{\rmi}{\eps} \parentheses{ \brackets{ \hat{x}, \hat{H} }
    + \hat{H} (\hat{x} - q(t)) } \mathcal{Z}_{0}(t)
    + \rmi\,\eps^{1/2}  (\hat{x} - q(t)) \zeta_{0}(t)\\
  &= (\hat{p} - p(t)) \mathcal{Z}_{0}(t)
    - \frac{\rmi}{\eps} \hat{H} (\hat{x} - q(t)) \mathcal{Z}_{0}(t)
    + \rmi\,\eps\,\hat{\xi}(t) \zeta_{0}(t),
\end{align*}
where the last equality follows from
\begin{align*}
  \brackets{ \hat{x}, \hat{H} }
  = \brackets{ \hat{x}, \frac{\hat{p}^{2}}{2} }
  = \rmi\,\eps\,\hat{p},
\end{align*}
and also setting $\hat{\xi}(t) \defeq \eps^{-1/2}(\hat{x} - q(t))$.
Applying $e^{\rmi\hat{H} t/\eps}$ to both sides, we have
\begin{multline*}
  e^{\rmi\hat{H} t/\eps} \od{}{t} \parentheses{ (\hat{x} - q(t)) \mathcal{Z}_{0}(t) }
  + e^{\rmi\hat{H} t/\eps} \frac{\rmi}{\eps} \hat{H} (\hat{x} - q(t)) \mathcal{Z}_{0}(t)
  = e^{\rmi\hat{H} t/\eps}(\hat{p} - p(t)) \mathcal{Z}_{0}(t)
  + \rmi\,\eps\,e^{\rmi\hat{H} t/\eps}\, \hat{\xi}(t) \zeta_{0}(t)
\end{multline*}
or
\begin{equation*}
  \od{}{t} \parentheses{ e^{\rmi\hat{H} t/\eps} (\hat{x} - q(t)) \mathcal{Z}_{0}(t) }
  = e^{\rmi\hat{H} t/\eps} (\hat{p} - p(t)) \mathcal{Z}_{0}(t)
  + \rmi\,\eps\, e^{\rmi\hat{H} t/\eps} \hat{\xi}(t) \zeta_{0}(t).
\end{equation*}
Integrating both sides on the interval $[0,t]$ and using $\mathcal{Z}_{0}(0) = 0$, we have
\begin{equation*}
  e^{\rmi\hat{H} t/\eps} (\hat{x} - q(t)) \mathcal{Z}_{0}(t)
  = \int_{0}^{t} e^{\rmi\hat{H} s/\eps} (\hat{p} - p(s)) \mathcal{Z}_{0}(s)\, ds
    + \rmi\,\eps \int_{0}^{t} e^{\rmi\hat{H} s/\eps}\, \hat{\xi}(s) \zeta_{0}(s)\, ds.
\end{equation*}
Taking the $L^{2}$-norm of the $i$-th component of both sides with $i \in \{1, \dots, d\}$,
\begin{align}
  \label{eq:x-var_ineq}
  \norm{ (\hat{x} - q(t))_{i} \mathcal{Z}_{0}(t) }
  &\le \int_{0}^{t} \norm{ (\hat{p} - p(s))_{i} \mathcal{Z}_{0}(s) }\, ds
    + \eps \int_{0}^{t} \norm{ \hat{\xi}_{i}(s) \zeta_{0}(s) }\, ds \nonumber\\
  &= \int_{0}^{t} \norm{ (\hat{p} - p(s))_{i} \mathcal{Z}_{0}(s) }\, ds
    + O(\eps),
\end{align}
where we used the estimate $\norm{ \hat{\xi}_{i}(s) \zeta_{0}(s) } = O(1)$ from Lemma~\ref{lem:zeta_0-estimate}~(\ref{lem:zeta_0-estimate4}).

Similarly,
\begin{align*}
  \od{}{t} \parentheses{ (\hat{p} - p(t)) \mathcal{Z}_{0}(t) }
  &= -\dot{p}(t)\, \mathcal{Z}_{0}(t) + (\hat{p} - p(t)) \dot{\mathcal{Z}}_{0}(t) \\
  &= \parentheses{ DV(q(t)) + \eps\,\partial_{q}V^{(1)}(q(t),Q(t)) } \mathcal{Z}_{0}(t)
    - \frac{\rmi}{\eps} (\hat{p} - p(t)) \hat{H} \mathcal{Z}_{0}(t) \\
  &\quad+ \rmi\,\eps^{1/2}  (\hat{p} - p(t)) \zeta_{0}(t) \\
  &= \parentheses{ DV(q(t)) + \eps\,\partial_{q}V^{(1)}(q(t),Q(t)) } \mathcal{Z}_{0}(t)
    - \frac{\rmi}{\eps} \parentheses{ \brackets{ \hat{p}, \hat{H} }
    + \hat{H} (\hat{p} - p(t)) } \mathcal{Z}_{0}(t) \\
  &\quad + \rmi\,\eps^{1/2}  (\hat{p} - p(t)) \zeta_{0}(t) \\
  &= \parentheses{ DV(q(t)) - DV(x) + \eps\,\partial_{q}V^{(1)}(q(t),Q(t)) } \mathcal{Z}_{0}(t)
    - \frac{\rmi}{\eps} \hat{H} (\hat{p} - p(t)) \mathcal{Z}_{0}(t) \\
  &\quad + \rmi\,\eps^{1/2}  (\hat{p} - p(t)) \zeta_{0}(t) \\
  &= \parentheses{ -D^{2}V(\sigma_{2}(x,q(t))) (\hat{x} - q(t)) + \eps\,\partial_{q}V^{(1)}(q(t),Q(t)) } \mathcal{Z}_{0}(t) \\
  &\quad - \frac{\rmi}{\eps} \hat{H} (\hat{p} - p(t)) \mathcal{Z}_{0}(t)
  + \rmi\,\eps\, \hat{\eta}(t) \zeta_{0}(t),
\end{align*}
where the second last equality follows from
\begin{align*}
  \brackets{ \hat{p}, \hat{H} }
  = \brackets{ \hat{p}, V(x) }
  = -\rmi\,\eps\,DV(x),
\end{align*}
and $\sigma_{2}(x,q(t))$ is a point in the segment joining $x$ and $q(t)$ in $\R^{d}$; we also set $\hat{\eta}(t) \defeq \eps^{-1/2}(\hat{p} - p(t))$.
Applying $e^{\rmi\hat{H} t/\eps}$ to both sides,
\begin{multline*}
  \od{}{t} \parentheses{ e^{\rmi\hat{H} t/\eps}  (\hat{p} - p(t)) \mathcal{Z}_{0}(t) } \\
  = e^{\rmi\hat{H} t/\eps} \parentheses{ -D^{2}V(\sigma_{2}(x,q(t))) (\hat{x} - q(t)) + \eps\,\partial_{q}V^{(1)}(q(t),Q(t)) } \mathcal{Z}_{0}(t) 
  + \rmi\,\eps\, e^{\rmi\hat{H} t/\eps} \hat{\eta}(t) \zeta_{0}(t).
\end{multline*}
Integrating both sides on $[0,t]$, we have
\begin{align*}
  e^{\rmi\hat{H} t/\eps} (\hat{p} - p(t)) \mathcal{Z}_{0}(t)
  &= \int_{0}^{t} e^{\rmi\hat{H} s/\eps} \parentheses{ -D^{2}V(\sigma_{2}(x,q(s))) (\hat{x} - q(s)) + \eps\,\partial_{q}V^{(1)}(q(s),Q(s)) } \mathcal{Z}_{0}(s)\, ds \\
  &\quad + \rmi\,\eps \int_{0}^{t} e^{\rmi\hat{H} s/\eps} \hat{\eta}(s) \zeta_{0}(s)\, ds.
\end{align*}
Taking the $L^{2}$-norm of the $i$-th component of both sides for any $i \in \{1, \dots, d\}$,
\begin{align}
  \label{eq:p-var_ineq}
  \norm{ (\hat{p} - p(t))_{i} \mathcal{Z}_{0}(t) }
  &\le \int_{0}^{t} \norm{ \sum_{j=1}^{d} D^{2}_{ij}V(\sigma_{2}(x,q(s))) (\hat{x} - q(s))_{j} \mathcal{Z}_{0}(s) }\,ds \nonumber\\
  &\quad + \eps \int_{0}^{t} \parentheses{
    \abs{ \partial_{q_{i}}V^{(1)}(q(s),Q(s)) }\,\norm{ \mathcal{Z}_{0}(s) }
    + \norm{ \hat{\eta}_{i}(s) \zeta_{0}(s) }
    }\, ds \nonumber\\
  &\le \int_{0}^{t} \sum_{j=1}^{d} \norm{ D^{2}_{ij}V(\sigma_{2}(x,q(s))) (\hat{x} - q(s))_{j} \mathcal{Z}_{0}(s) }\,ds
    + O(\eps) \nonumber\\
  &\le C_{3} \int_{0}^{t} \sum_{j=1}^{d} \norm{ (\hat{x} - q(s))_{j} \mathcal{Z}_{0}(s) }\,ds
    + O(\eps),
\end{align}
where we used the following bound of the second derivative of $V$
\begin{equation*}
  C_{3} \defeq \max_{1\le i,j\le d} \sup_{x \in \R^{d}} \abs{ D^{2}_{ij}V(x) }
\end{equation*}
as well as the following: $\norm{ \mathcal{Z}_{0}(s) } = O(\eps^{1/2})$ from Corollary~\ref{cor:magic_formula} and $\norm{ \hat{\eta}_{i}(s) \zeta_{0}(s) } = O(1)$ from Lemma~\ref{lem:zeta_0-estimate}~(\ref{lem:zeta_0-estimate5}).

Now, let us set
\begin{equation*}
  f(t) \defeq \sum_{i=1}^{d} \parentheses{ \norm{ (\hat{x} - q(t))_{i} \mathcal{Z}_{0}(t) } + \norm{ (\hat{p} - p(t))_{i} \mathcal{Z}_{0}(t) } }.
\end{equation*}
Then, using \eqref{eq:x-var_ineq} and \eqref{eq:p-var_ineq}, we have
\begin{align*}
  f(t) &\le \int_{0}^{t} \sum_{i=1}^{d} \norm{ (\hat{p} - p(s))_{i} \mathcal{Z}_{0}(s) }\,ds
         + d\,C_{3} \int_{0}^{t} \sum_{j=1}^{d} \norm{ (\hat{x} - q(s))_{j} \mathcal{Z}_{0}(s) }\,ds
         + O(\eps) \\
       &\le C_{4} \int_{0}^{t} \sum_{i=1}^{d} \parentheses{
         \norm{ (\hat{p} - p(s))_{i} \mathcal{Z}_{0}(s) } + \norm{ (\hat{x} - q(s))_{i} \mathcal{Z}_{0}(s) }
         }\, ds
         + O(\eps) \\
       &= C_{4} \int_{0}^{t} f(s)\,ds + O(\eps),
\end{align*}
where we defined $C_{4} \defeq \max\{ 1, d\,C_{3} \}$.
Therefore, by Gronwall's inequality~\cite{Gr1919}, we obtain
\begin{equation*}
  f(t) \le O(\eps) \exp(C_{4}t),
\end{equation*}
that is, $f(t) = O(\eps)$, and so we have, for any $i \in \{1, \dots, d\}$,
\begin{equation*}
  \norm{ (\hat{x} - q(t))_{i} \mathcal{Z}_{0}(t) } = O(\eps),
  \qquad
  \norm{ (\hat{p} - p(t))_{i} \mathcal{Z}_{0}(t) } = O(\eps).
\end{equation*}

As a result, by the Cauchy--Schwarz inequality, we obtain
\begin{equation*}
  \abs{ \ip{\mathcal{Z}_{0}(t)}{(\hat{x} - q(t))_{i} \mathcal{Z}_{0}(t)} }
  \le \norm{ \mathcal{Z}_{0}(t) }\, \norm{ (\hat{x} - q(t))_{i} \mathcal{Z}_{0}(t) } = O(\eps^{3/2}),
\end{equation*}
and similarly
\begin{equation*}
  \abs{ \ip{\mathcal{Z}_{0}(t)}{(\hat{p} - p(t))_{i} \mathcal{Z}_{0}(t)} } = O(\eps^{3/2})
\end{equation*}
as well.

Therefore, we conclude that, for any $i \in \{1, \dots, d\}$,
\begin{equation*}
  \exval{\hat{x}_{i}}(t) - q_{i}(t) = O(\eps^{3/2}),
  \qquad
  \exval{\hat{p}_{i}}(t) - p_{i}(t) = O(\eps^{3/2}).
\end{equation*}

\begin{remark}
  If one uses the classical Hamiltonian system for $(q,p)$ as in \eqref{eq:Hagedorn} of Hagedorn, then the estimate of the second error term proceeds similarly and hence it is still $O(\eps^{3/2})$---the difference is the expression of the residual term $\zeta_{0}$ as well as the absence of the term with $V^{(1)}$.
  These do not affect the estimate of the second error term---it is still $O(\eps^{3/2})$.
  However, as discussed in Remark~\ref{rem:Hagedorn2}, the estimate of the first error term now becomes $O(\eps)$ and hence the total error is $O(\eps)$.
\end{remark}

\section*{Acknowledgments}
I would like to thank George Hagedorn, Caroline Lasser, Alex Watson, and Bin Cheng for helpful comments and discussions.
This work was partially supported by NSF grant DMS-2006736.
The author states that there is no conflict of interest.

\appendix
\section{Additional Details}
\subsection{Details of Proof of Lemma~\ref{lem:zeta_0-estimate}~(\ref{lem:zeta_0-estimate5})}
\label{sseca:zeta_0-estimate5}
We suppress the time dependence of $\zeta_{0}$, $\phi_{0}$, and $(q, p, Q, P)$ for brevity here.
First we have
\begin{equation*}
  (\hat{p}_{i} - p_{i}) \zeta_{0}(x)
  = \parentheses{ \hat{p}_{i} \alpha(q,Q; x) } \phi_{0}(x)
  + \alpha(q,Q; x) \parentheses{ \hat{p}_{i} \phi_{0}(x) }
  - p_{i} \zeta_{0}(x).
\end{equation*}
However,
\begin{align*}
  \hat{p}_{i} \alpha(q,Q; x)
  &= -\rmi\,\eps \pd{}{x_{i}} \alpha(q,Q; x) \\
  &= \eps^{1/2} \beta_{i}(q,Q;x) \\
  &= \eps^{1/2} \parentheses{ \beta_{i}^{(0)}(q,Q;x) + \eps^{1/2} \beta_{i}^{(1)}(q;x) },
\end{align*}
where we defined
\begin{gather*}
  \beta_{i}(q,Q; x) \defeq \beta_{i}^{(0)}(q,Q;x) + \eps^{1/2}\beta_{i}^{(1)}(q;x), \\
  \beta_{i}^{(0)}(q, Q; x) \defeq -\rmi\,\eps^{1/2} \pd{}{x_{i}} \alpha^{(0)}(q,Q; x),
  \qquad
  \beta_{i}^{(1)}(q; x) \defeq -\rmi\,\eps^{1/2} \pd{}{x_{i}} \alpha^{(1)}(q; x),
\end{gather*}
which yield the expressions in \eqref{eq:betas}.
On the other hand, using the expression \eqref{eq:phi_0} for $\phi_{0}$,
\begin{equation*}
  \hat{p} \phi_{0}(x) = \parentheses{ P Q^{-1}(x - q) + p } \phi_{0}(x).
\end{equation*}
Therefore,
\begin{equation*}
  (\hat{p}_{i} - p_{i}) \zeta_{0}(x)
  = \parentheses{ \eps^{1/2}\beta_{i}(q, Q; x) + \alpha(q,Q;x)  (P Q^{-1})_{ij}(x - q)_{j} } \phi_{0}(x),
\end{equation*}
and thus
\begin{equation*}
  \hat{\eta}_{i} \zeta_{0}(x)
  = \parentheses{ \beta_{i}(q,Q; x) + \alpha(q,Q;x)  (P Q^{-1})_{ij} \xi_{j} } \phi_{0}(x),
\end{equation*}
which gives \eqref{eq:eta_zeta}.

\subsection{Details of Proof of Lemma~\ref{lem:orthogonality}}
\label{sseca:orthogonality}
Let us show the detailed derivation of \eqref{eq:alpha0-2nd_term}.
We first have, using \eqref{eq:hatx-q_in_ladder_ops},
\begin{multline*}
  \frac{\eps^{-3/2}}{3!} D^{3}_{ijk}V(q) (x - q)^{3}_{ijk} \phi_{0}
  = \frac{1}{12\sqrt{2}} \mathcal{T}_{ijk} (\overline{Q}_{il} \mathscr{A}_{l} + Q_{il} \mathscr{A}^{*}_{l}) (\overline{Q}_{jm} \mathscr{A}_{m} + Q_{jm} \mathscr{A}^{*}_{m}) (\overline{Q}_{kn} \mathscr{A}_{n} + Q_{kn} \mathscr{A}^{*}_{n}) \phi_{0}.
\end{multline*}
Notice that, applying lowering operator(s) twice and a raising operator once---regardless of the order---to $\phi_{0}$ results in zero, and the same goes with lowering operator(s) thrice as well.
Therefore, we have
\begin{align*}
  (\overline{Q}_{il} \mathscr{A}_{l}& + Q_{il} \mathscr{A}^{*}_{l}) (\overline{Q}_{jm} \mathscr{A}_{m} + Q_{jm} \mathscr{A}^{*}_{m}) (\overline{Q}_{kn} \mathscr{A}_{n} + Q_{kn} \mathscr{A}^{*}_{n}) \phi_{0} \\
                                    &= 
                                      Q_{il} Q_{jm} Q_{kn} \mathscr{A}^{*}_{l} \mathscr{A}^{*}_{m} \mathscr{A}^{*}_{n} \phi_{0} \\
                                    &\quad + \parentheses{
                                      \overline{Q}_{il} Q_{jm} Q_{kn} \mathscr{A}_{l} \mathscr{A}^{*}_{m} \mathscr{A}^{*}_{n}
                                      + Q_{il} \overline{Q}_{jm} Q_{kn} \mathscr{A}^{*}_{l} \mathscr{A}_{m} \mathscr{A}^{*}_{n}
                                      + Q_{il} Q_{jm} \overline{Q}_{kn} \mathscr{A}^{*}_{l} \mathscr{A}^{*}_{m} \mathscr{A}_{n}
                                      } \phi_{0} \\
                                    &= Q_{il} Q_{jm} Q_{kn} \mathscr{A}^{*}_{l} \mathscr{A}^{*}_{m} \mathscr{A}^{*}_{n} \phi_{0}
                                      + \parentheses{
                                      \overline{Q}_{il} Q_{jm} Q_{kn} \mathscr{A}_{l} \mathscr{A}^{*}_{m} \mathscr{A}^{*}_{n}
                                      + Q_{il} \overline{Q}_{jm} Q_{kn} \mathscr{A}^{*}_{l} \mathscr{A}_{m} \mathscr{A}^{*}_{n}
                                      } \phi_{0},
\end{align*}
where we used \eqref{eq:A-phi_0}.

However, we may use \eqref{eq:commutator-As} and \eqref{eq:A-phi_0} to simplify the second and last terms as follows:
\begin{align*}
  \mathscr{A}_{l} \mathscr{A}^{*}_{m} \mathscr{A}^{*}_{n} \phi_{0}
  &= (\delta_{lm} + \mathscr{A}^{*}_{m} \mathscr{A}_{l}) \mathscr{A}^{*}_{n} \phi_{0} \\
  &= \delta_{lm} \mathscr{A}^{*}_{n} \phi_{0} + \mathscr{A}^{*}_{m} (\delta_{ln} +  \mathscr{A}^{*}_{n} \mathscr{A}_{l}) \phi_{0} \\
  &= (\delta_{lm} \mathscr{A}^{*}_{n} + \delta_{ln} \mathscr{A}^{*}_{m}) \phi_{0},
\end{align*}
and
\begin{align*}
  \mathscr{A}^{*}_{l} \mathscr{A}_{m} \mathscr{A}^{*}_{n} \phi_{0}
  &= \mathscr{A}^{*}_{l} (\delta_{mn} + \mathscr{A}^{*}_{n} \mathscr{A}_{m}) \phi_{0} \\
  &= \delta_{mn} \mathscr{A}^{*}_{l} \phi_{0}.
\end{align*}
Therefore,
\begin{align*}
  (\overline{Q}_{il} \mathscr{A}_{l}& + Q_{il} \mathscr{A}^{*}_{l}) (\overline{Q}_{jm} \mathscr{A}_{m} + Q_{jm} \mathscr{A}^{*}_{m}) (\overline{Q}_{kn} \mathscr{A}_{n} + Q_{kn} \mathscr{A}^{*}_{n}) \phi_{0} \\
                                    &= Q_{il} Q_{jm} Q_{kn} \mathscr{A}^{*}_{l} \mathscr{A}^{*}_{m} \mathscr{A}^{*}_{n} \phi_{0}
                                      + \parentheses{
                                      \overline{Q}_{il} Q_{jl} Q_{kn} \mathscr{A}^{*}_{n}
                                      + \overline{Q}_{il} Q_{jm} Q_{kl} \mathscr{A}^{*}_{m}
                                      + Q_{il} \overline{Q}_{jm} Q_{km} \mathscr{A}^{*}_{l} 
                                      } \phi_{0},
\end{align*}
and so
\begin{align*}
  \frac{\eps^{-3/2}}{3!} D^{3}_{ijk}V(q) (x - q)^{3}_{ijk} \phi_{0}
  &= \frac{1}{12\sqrt{2}} \mathcal{T}_{ijk} 
    Q_{il} Q_{jm} Q_{kn} \mathscr{A}^{*}_{l} \mathscr{A}^{*}_{m} \mathscr{A}^{*}_{n} \phi_{0} \\
  &\quad + \frac{1}{12\sqrt{2}}  \mathcal{T}_{ijk} \parentheses{
    \overline{Q}_{il} Q_{jl} Q_{kn} \mathscr{A}^{*}_{n}
    + \overline{Q}_{il} Q_{jm} Q_{kl} \mathscr{A}^{*}_{m}
    + Q_{il} \overline{Q}_{jm} Q_{km} \mathscr{A}^{*}_{l} 
    } \phi_{0} \\
  &= \frac{1}{4\sqrt{3}} \mathcal{T}_{ijk} 
    Q_{il} Q_{jm} Q_{kn} \phi_{\mathbf{e}_{l} + \mathbf{e}_{m} + \mathbf{e}_{n}}
    + \frac{1}{4\sqrt{2}} \mathcal{T}_{ijk} 
    \overline{Q}_{il} Q_{jl} Q_{kn} \phi_{\mathbf{e}_{n}},
\end{align*}
where we used \eqref{eq:phi_n-raised} to rewrite $\mathscr{A}^{*}_{l} \mathscr{A}^{*}_{m} \mathscr{A}^{*}_{n} \phi_{0}$ as $\sqrt{6}\,\phi_{\mathbf{e}_{l} + \mathbf{e}_{m} + \mathbf{e}_{n}}$ as well as the permutation symmetry of $\mathcal{T}$ in its indices.
Hence we obtain \eqref{eq:alpha0-2nd_term}.

\subsection{Details on Remark~\ref{rem:Hagedorn2}}
\label{sseca:Hagedorn2}
We drop the time and spatial dependence for brevity here.
Just as we have done in the above subsection, rewriting the cubic term in \eqref{eq:alpha-Hagedorn} for $\alpha$ using \eqref{eq:hatx-q_in_ladder_ops}, we obtain
\begin{align*}
  \zeta_{0} = \alpha(q, Q)\,\phi_{0}
  &= -\frac{1}{4\sqrt{3}} \mathcal{T}_{abc} Q_{al} Q_{bm} Q_{cn} \phi_{\mathbf{e}_{l} + \mathbf{e}_{m} + \mathbf{e}_{n}}
    - \frac{1}{4\sqrt{2}} \mathcal{T}_{abc} \overline{Q}_{al} Q_{bl} Q_{cn} \phi_{\mathbf{e}_{n}} \\
  &\quad - \eps^{1/2}\,\frac{1}{4!} D^{4}V(\sigma_{1}(x,q)) \cdot \xi^{4} \phi_{0},
\end{align*}
where $\mathcal{T}_{abc} \defeq D^{3}_{abc}V(q)$ and $\xi \defeq \eps^{-1/2}(x - q)$.
Therefore,
\begin{align}
  \ip{ \zeta_{0} }{ \phi_{\mathbf{e}_{j}} }
  &= -\frac{1}{4\sqrt{3}} \mathcal{T}_{abc} \overline{Q}_{al} \overline{Q}_{bm} \overline{Q}_{cn} \ip{ \phi_{\mathbf{e}_{l} + \mathbf{e}_{m} + \mathbf{e}_{n}} }{ \phi_{\mathbf{e}_{j}} }
    - \frac{1}{4\sqrt{2}} \mathcal{T}_{abc} Q_{al} \overline{Q}_{bl} \overline{Q}_{cn} \ip{ \phi_{\mathbf{e}_{n}} }{ \phi_{\mathbf{e}_{j}} } \nonumber\\
  &\quad - \eps^{1/2}\,\frac{1}{4!} \ip{ D^{4}V(\sigma_{1}(x,q)) \cdot \xi^{4} \phi_{0} }{ \phi_{\mathbf{e}_{j}} } \nonumber\\
  &= - \frac{1}{4\sqrt{2}} \mathcal{T}_{abc} Q_{al} \overline{Q}_{bl} \overline{Q}_{cj} + O(\eps^{1/2}) \nonumber\\
  &= - \frac{1}{\sqrt{2}} Q^{*}_{jc} \partial_{q_{c}}V^{(1)}(q,Q) + O(\eps^{1/2}),
    \label{eq:ip-Hagedorn2}
\end{align}
because we can obtain the estimate
\begin{equation*}
  \frac{1}{4!} \ip{ D^{4}V(\sigma_{1}(x,q)) \cdot \xi^{4} \phi_{0} }{ \phi_{\mathbf{e}_{j}} } = O(1)
\end{equation*}
just as we did in the proof of Lemma~\ref{lem:zeta_0-estimate}~(\ref{lem:zeta_0-estimate2}).
We also used the expression~\eqref{eq:V^1} for $V^{(1)}$ in the last equality.
As a result, we have $\ip{ \zeta_{0} }{ \phi_{\mathbf{e}_{j}} } = O(1)$.
    
\bibliography{ExVal-SympGWP}

\begin{thebibliography}{34}
\providecommand{\natexlab}[1]{#1}
\providecommand{\url}[1]{\texttt{#1}}
\expandafter\ifx\csname urlstyle\endcsname\relax
  \providecommand{\doi}[1]{doi: #1}\else
  \providecommand{\doi}{doi: \begingroup \urlstyle{rm}\Url}\fi

\bibitem[Bouzouina and Robert(2002)]{BoRo2002}
A.~Bouzouina and D.~Robert.
\newblock Uniform semiclassical estimates for the propagation of quantum
  observables.
\newblock \emph{Duke Math. J.}, 111\penalty0 (2):\penalty0 223--252, 2002.

\bibitem[Combescure and Robert(2012)]{CoRo2012}
M.~Combescure and D.~Robert.
\newblock \emph{Coherent States and Applications in Mathematical Physics}.
\newblock Springer, 2012.

\bibitem[Egorov(1969)]{Eg1969}
Y.~V. Egorov.
\newblock The canonical transformations of pseudodifferential operators.
\newblock \emph{Uspekhi Mat.~Nauk}, 24\penalty0 (5(149)):\penalty0 235--236,
  1969.

\bibitem[Faou and Lubich(2006)]{FaLu2006}
E.~Faou and C.~Lubich.
\newblock A {P}oisson integrator for {G}aussian wavepacket dynamics.
\newblock \emph{Computing and Visualization in Science}, 9\penalty0
  (2):\penalty0 45--55, 2006.

\bibitem[Gronwall(1919)]{Gr1919}
T.~H. Gronwall.
\newblock Note on the derivatives with respect to a parameter of the solutions
  of a system of differential equations.
\newblock 20\penalty0 (4):\penalty0 292--296, 1919.

\bibitem[Hagedorn(1980)]{Ha1980}
G.~A. Hagedorn.
\newblock Semiclassical quantum mechanics. {I}. {T}he $\hbar\to0$ limit for
  coherent states.
\newblock \emph{Communications in Mathematical Physics}, 71\penalty0
  (1):\penalty0 77--93, 1980.

\bibitem[Hagedorn(1981)]{Ha1981}
G.~A. Hagedorn.
\newblock Semiclassical quantum mechanics. {III}. the large order asymptotics
  and more general states.
\newblock \emph{Annals of Physics}, 135\penalty0 (1):\penalty0 58--70, 1981.

\bibitem[Hagedorn(1985)]{Hagedorn1985}
G.~A. Hagedorn.
\newblock Semiclassical quantum mechanics, {IV}: large order asymptotics and
  more general states in more than one dimension.
\newblock \emph{Annales de l'institut Henri Poincar{\'e} (A) Physique
  th{\'e}orique}, 42\penalty0 (4):\penalty0 363--374, 1985.

\bibitem[Hagedorn(1998)]{Ha1998}
G.~A. Hagedorn.
\newblock Raising and lowering operators for semiclassical wave packets.
\newblock \emph{Annals of Physics}, 269\penalty0 (1):\penalty0 77--104, 1998.

\bibitem[Hagedorn and Joye(1999)]{HaJo1999}
G.~A. Hagedorn and A.~Joye.
\newblock Semiclassical dynamics with exponentially small error estimates.
\newblock \emph{Communications in Mathematical Physics}, 207\penalty0
  (2):\penalty0 439--465, 1999.

\bibitem[Hagedorn and Joye(2000)]{HaJo2000}
G.~A. Hagedorn and A.~Joye.
\newblock Exponentially accurate semiclassical dynamics: Propagation,
  localization, ehrenfest times, scattering, and more general states.
\newblock \emph{Annales Henri Poincar{\'e}}, 1\penalty0 (5):\penalty0 837--883,
  2000.

\bibitem[Heller(1975)]{He1975a}
E.~J. Heller.
\newblock Time-dependent approach to semiclassical dynamics.
\newblock \emph{Journal of Chemical Physics}, 62\penalty0 (4):\penalty0
  1544--1555, 1975.

\bibitem[Heller(1976)]{He1976b}
E.~J. Heller.
\newblock Classical {$S$}-matrix limit of wave packet dynamics.
\newblock \emph{Journal of Chemical Physics}, 65\penalty0 (11):\penalty0
  4979--4989, 1976.

\bibitem[Heller(1991)]{Heller-LesHouches}
E.~J. Heller.
\newblock Wavepacket dynamics and quantum chaology.
\newblock In M.~Giannoni, A.~Voros, and J.~Zinn-Justin, editors, \emph{Chaos
  and quantum physics}, pages 547--663. North-Holland, 1991.

\bibitem[Lasser and R{\"o}blitz(2010)]{LaRo2010}
C.~Lasser and S.~R{\"o}blitz.
\newblock Computing expectation values for molecular quantum dynamics.
\newblock \emph{SIAM Journal on Scientific Computing}, 32\penalty0
  (3):\penalty0 1465--1483, 2010.

\bibitem[Lasser and Lubich(2020)]{LaLu2020}
C.~Lasser and C.~Lubich.
\newblock Computing quantum dynamics in the semiclassical regime.
\newblock \emph{Acta Numerica}, 29:\penalty0 229--401, 2020.

\bibitem[Lubich(2008)]{Lu2008}
C.~Lubich.
\newblock \emph{From quantum to classical molecular dynamics: reduced models
  and numerical analysis}.
\newblock European Mathematical Society, Z{\"u}rich, Switzerland, 2008.

\bibitem[Miller(2006)]{Mi2006}
P.~D. Miller.
\newblock \emph{Applied Asymptotic Analysis}.
\newblock American Mathematical Society, Providence, R.I., 2006.

\bibitem[Miller(1970)]{Mi1970}
W.~H. Miller.
\newblock Classical {$S$} matrix: Numerical application to inelastic
  collisions.
\newblock \emph{The Journal of Chemical Physics}, 53\penalty0 (9):\penalty0
  3578--3587, 1970.

\bibitem[Miller(1974)]{Mi1974b}
W.~H. Miller.
\newblock Quantum mechanical transition state theory and a new semiclassical
  model for reaction rate constants.
\newblock \emph{The Journal of Chemical Physics}, 61\penalty0 (5):\penalty0
  1823--1834, 1974.

\bibitem[Miller(2001)]{Mi2001}
W.~H. Miller.
\newblock The semiclassical initial value representation: A potentially
  practical way for adding quantum effects to classical molecular dynamics
  simulations.
\newblock \emph{The Journal of Physical Chemistry A}, 105\penalty0
  (13):\penalty0 2942--2955, 2001.

\bibitem[Ohsawa(2015{\natexlab{a}})]{Oh2015b}
T.~Ohsawa.
\newblock Symmetry and conservation laws in semiclassical wave packet dynamics.
\newblock \emph{Journal of Mathematical Physics}, 56\penalty0 (3):\penalty0
  032103, 2015{\natexlab{a}}.

\bibitem[Ohsawa(2015{\natexlab{b}})]{Oh2015c}
T.~Ohsawa.
\newblock The {S}iegel upper half space is a {M}arsden--{W}einstein quotient:
  Symplectic reduction and {G}aussian wave packets.
\newblock \emph{Letters in Mathematical Physics}, 105\penalty0 (9):\penalty0
  1301--1320, 2015{\natexlab{b}}.

\bibitem[Ohsawa(2019)]{Oh2019b}
T.~Ohsawa.
\newblock The {H}agedorn--{H}ermite correspondence.
\newblock \emph{Journal of Fourier Analysis and Applications}, 25\penalty0
  (4):\penalty0 1513--1552, 2019.

\bibitem[Ohsawa and Leok(2013)]{OhLe2013}
T.~Ohsawa and M.~Leok.
\newblock Symplectic semiclassical wave packet dynamics.
\newblock \emph{Journal of Physics A: Mathematical and Theoretical},
  46\penalty0 (40):\penalty0 405201, 2013.

\bibitem[Ohsawa and Tronci(2017)]{OhTr2017}
T.~Ohsawa and C.~Tronci.
\newblock Geometry and dynamics of {G}aussian wave packets and their {W}igner
  transforms.
\newblock \emph{Journal of Mathematical Physics}, 58\penalty0 (9):\penalty0
  092105, 2017.

\bibitem[Pattanayak and Schieve(1994)]{PaSc1994}
A.~K. Pattanayak and W.~C. Schieve.
\newblock Gaussian wave-packet dynamics: Semiquantal and semiclassical
  phase-space formalism.
\newblock \emph{Physical Review E}, 50\penalty0 (5):\penalty0 3601--3615, 1994.

\bibitem[Prezhdo(2006)]{Pr2006}
O.~V. Prezhdo.
\newblock Quantized {H}amilton dynamics.
\newblock \emph{Theoretical Chemistry Accounts}, 116\penalty0 (1-3):\penalty0
  206--218, 2006.

\bibitem[Prezhdo and Pereverzev(2000)]{PrPe2000}
O.~V. Prezhdo and Y.~V. Pereverzev.
\newblock Quantized {H}amilton dynamics.
\newblock \emph{Journal of Chemical Physics}, 113\penalty0 (16):\penalty0
  6557--6565, 2000.

\bibitem[Prezhdo and Pereverzev(2002)]{PrPe2002}
O.~V. Prezhdo and Y.~V. Pereverzev.
\newblock Quantized {H}amilton dynamics for a general potential.
\newblock \emph{Journal of Chemical Physics}, 116\penalty0 (11):\penalty0
  4450--4461, 2002.

\bibitem[Robert(2007)]{Ro2007}
D.~Robert.
\newblock Propagation of coherent states in quantum mechanics and applications.
\newblock In X.~Wang, editor, \emph{Partial differential equations and
  applications}, volume~15 of \emph{S{\'e}minaires et Congr{\`e}s}, pages
  181--252. Soci{\'e}t{\'e} Math{\'e}matique de France, 2007.

\bibitem[Wang et~al.(1998)Wang, Sun, and Miller]{WaSuMi1998}
H.~Wang, X.~Sun, and W.~H. Miller.
\newblock Semiclassical approximations for the calculation of thermal rate
  constants for chemical reactions in complex molecular systems.
\newblock \emph{The Journal of Chemical Physics}, 108\penalty0 (23):\penalty0
  9726--9736, 1998.

\bibitem[Watson et~al.(2017)Watson, Lu, and Weinstein]{WaLuWe2017}
A.~B. Watson, J.~Lu, and M.~I. Weinstein.
\newblock Wavepackets in inhomogeneous periodic media: Effective particle-field
  dynamics and {B}erry curvature.
\newblock \emph{Journal of Mathematical Physics}, 58\penalty0 (2):\penalty0
  021503, 2017.

\bibitem[Zworski(2012)]{Zw2012}
M.~Zworski.
\newblock \emph{Semiclassical Analysis}.
\newblock American Mathematical Society, Providence, R.I., 2012.

\end{thebibliography}
\bibliographystyle{plainnat}

\end{document}